\documentclass{llncs}

\usepackage{fullpage}

\usepackage{amsmath}
\usepackage{amssymb}
\usepackage{amsfonts}
\usepackage{amstext}
\usepackage{xspace}
\usepackage{graphicx}

\usepackage{mathrsfs} 
\usepackage{boxedminipage}
\usepackage{textcomp}
\usepackage{pifont}

\usepackage{adjustbox}
\usepackage{bm}

\usepackage{tikz}
\usetikzlibrary{shapes,positioning}
\usetikzlibrary{snakes}
\usetikzlibrary{patterns}
\usetikzlibrary{calc}
\usetikzlibrary{arrows,shapes}
\pgfdeclarelayer{background}
\pgfsetlayers{background,main}

\tikzstyle{vertex}=[circle,fill=white,draw=black,line width=2.5pt,minimum size=36pt,inner sep=.5ex,font=\huge]
\tikzstyle{selected vertex} = [vertex, fill=none,draw=red,line width=4pt,font=\huge]
\tikzstyle{edge} = [draw,line width=2pt,-]
\tikzstyle{weight} = [font=\large]
\tikzstyle{selected edge} = [draw,line width=4pt,-,blue]
\tikzstyle{ignored edge} = [draw,line width=2pt,-,blue]

\usetikzlibrary{matrix,decorations.pathmorphing}

\usepackage{subfigure}

\usepackage[boxed,ruled,vlined,linesnumbered]{algorithm2e}

\newcommand{\eat}[1]{}

\newcounter{lxmalgC}

\newcommand{\lxmNP}{\textsf{NP}}
\newcommand{\lxmPTime}{\textsf{P}}

\usepackage{colortbl} % color the table
\usepackage{tabularx} % specify the width of table column
\newcolumntype{L}[1]{>{\raggedright\arraybackslash}p{#1}}
\newcolumntype{C}[1]{>{\centering\arraybackslash}p{#1}}
\newcolumntype{R}[1]{>{\raggedleft\arraybackslash}p{#1}}
\newcommand{\lxmProblemTwoPartition}{{2BGP}}
\newcommand{\lxmProblemPrevious}{{$k$BGP}}
\newcommand{\lxmProblemPreviousDouble}{{$k{\tau}$BGP}}

\newcommand{\lxmProblemWordload}{W$k$BGP}
\newcommand{\lxmProblemMotif}{M$k$BGP}
\newcommand{\lxmProblemMotifTwo}{M$2$BGP}
\newcommand{\lxmProblemMotifTriangle}{$\mathrm{\Delta}-k$BGP}
\newcommand{\lxmProblemMotifTriangleTwoPar}{$\mathrm{\Delta}-$2BGP}

\newcommand{\lxmAlgorithmWorkload}{\textsc{WkBGPartition}}

\newcommand{\lxmAlgorithmMotif}{\textsc{Tri-kBGPartition}}

\title{Balanced Partitioning for Optimizing Big Graph Computation: Complexities and Approximation Algorithms}
{%%%%%%%%%%%%%%%%%%%%%%%%%%%%%%%%%%%%
	\author{Baoling Ning\inst{1} %\orcidID{0000-0002-3643-106X} 
		\and Jianzhong Li\inst{2,3}}
	\authorrunning{B. Ning et al.}
	\institute{Heilongjiang University, China. \email{ningbaoling2009@163.com} \and Harbin Institute of Technology, China. \and
		Shenzhen Institute of Advanced Technology, Chinese Academy of Sciences, China.
		}
}%%%%%%%%%%%%%%%%%%%%%%%%%%%%%%%%%%%%%%

\date{}
\begin{document}

\maketitle

\begin{abstract}
	
Graph partitioning is a key fundamental problem in the area of big graph computation.
Previous works do not consider the practical requirements when
optimizing the big data analysis in real applications.
In this paper, motivated by optimizing the big data computing applications,
two typical problems of graph partitioning  are studied.
The first problem is to optimize the performance of specific workloads
by graph partitioning, which lacks of algorithms with performance guarantees.
The second problem is to optimize the computation of motifs
by graph partitioning, which has not been focused by previous works.
First, the formal definitions of the above two problems are introduced,
and the semidefinite programming representations are also designed
based on the analysis of the properties of the two problems.
For the motif based partitioning problem,  it is proved to be $\lxmNP$-complete even for the special case of $k=2$ and the motif is a triangle,
and  its inapproximability is also shown by proving that there are no efficient algorithms with finite approximation ratio.
Finally, using the semidifinite programming and sophisticated rounding techniques, 
the bi-criteria  $O(\sqrt{\log n\log k})$-approximation algorithms with polynomial time
cost are designed and analyzed for them.

\keywords{Graph Partitioning \and Bi-criteria Approximation Algorithm \and Semidefinite Programming}
\end{abstract}

\section{Introduction}
\label{sec:introduction}

Graph partitioning is a widely used strategy when processing
big graph data. 
On one side, graph partitioning provides fundamental functions
for supporting many important graph analysis operations (\emph{e.g.}, \emph{traversal},
\emph{pagerank}, \emph{clustering}).
On the other side, graph partitioning is a key fundamental
problem in big graph computing platforms, which is the core of many mechanisms such as load
balancing.
Therefore, the problem of balanced graph partitioning is important in
the area of big graph analysis.

%In big graph computing platforms, the most commonly used graph partitioning
%methods are \textbf{xxx}.

In previous works, the most common problem focused by the them
is called $k$-balanced graph partitioning, {\lxmProblemPrevious} for short.
Given a graph $G$ and an integer $k\geq 2$, the {\lxmProblemPrevious} problem
is to find a solution for partitioning the graph nodes, such that all nodes are partitioned 
into $k$ disjoint partitions of the same size, and the partitioning cost is minimal.
Here, the partitioning cost should be chosen according to the practical requirements,
the most popular cost is the size (or weighted size) of the \emph{cut} edges between
different partitions.
The problem has been proved to be $\lxmNP$-complete, as one of the most
classical $\lxmNP$-complete problems found in the 1970s, more precisely,
even considering the special case of $k=2$, the {\lxmProblemPrevious} problem is still
$\lxmNP$-complete \cite{Garey76SomeSimplified}.
Then, a natural and important research issue is try to design an efficient approximation
solution, however, in 2006, it is proved by \cite{Andreev06BalancedGraph} that
for arbitrary $k\geq 3$ there are no polynomial-time algorithms for {\lxmProblemPrevious}
with finite approximation ratio unless \textsf{P}=\lxmNP.
Thus, it is almost impossible to design efficient approximation algorithms for \lxmProblemPrevious,
and the follow-up research works have no way but to relax the performance requirements
further. One popular and principal way is to design \emph{bi-criteria approximation algorithms},
where the \lxmProblemPrevious problem is relaxed to the $(k,\tau)$-balanced graph partition problem
({\lxmProblemPreviousDouble} for short).
Given a graph $G$, an integer $k\geq 2$ and a real number $\tau\geq 1$,
the {\lxmProblemPreviousDouble} problem is to find a solution for partitioning
the nodes into $k$ partitions of at most $\tau\lceil\frac{|V_G|}{k}\rceil$ size each,
such that the partitioning cost is minimal.
The meaning of \emph{bi-criteria} is that the output is relaxed to be a solution of the {\lxmProblemPreviousDouble} problem but
the cost is compared with the optimal solution of the original {\lxmProblemPrevious} problem (equivalently, $\tau=1$).
Current works on \emph{bi-criteria} approximation algorithms can be divided into three kinds.
(1) The first kind is designed based on the algorithms for the {\lxmProblemTwoPartition} problem, and
the best approximation ratio is $O(\log{k}\sqrt{\log{n}})$ \cite{Arora04ExpanderFlows,Simon97HowGood}.
(2) The second kind is designed using the mathematical programming techniques,
and the best ratio is $O(\sqrt{\log{n}\log{k}})$ \cite{Krauthgamer09PartitioningGraphs,Even99FastApproximate}.
(3) The third kind is designed based on combinatorial techniques such as dynamic programming,
and the best approximation ratio can be $O(\log{n})$ \cite{Feldmann15BalancedPartitions,Andreev06BalancedGraph,Racke08OptimalHierarchical}.

As shown above, previous works have focused on the graph partitioning problem and
produced a plenty of insightful results.
However, in view of the applications of big graph processing systems such as \cite{Martella17SpinnerScalable}
and \cite{Gonzalez12PowerGraph}, many novel requirements for the partitioning problem have emerged,
and current methods can not satisfy the new requirements and does not
support the optimization for novel processing of graph data.

Therefore, in this paper, two typical problems of graph partitioning motivated by
the big data computing applications are studied.
The first problem is to optimize the performance of specific workloads
by adjusting the partitioning of graph data.
Current works focusing on this problem aim to design heuristic \cite{Pacaci19ExperimentalAnalysis}
or learning based partitioning methods \cite{Fan20ApplicationDriven}, and one of their drawbacks
is the lack of performance guarantees.
Our goal is to provide efficient partitioning methods with performance guarantee
to optimize the computation of the given workloads.
The second problem is to optimize the computation of motifs
by adjusting the partitioning of graph data.
As far as we know, there still lacks of research works focusing on this problem.
Our goal is to provide the fundamental theoretical analysis results and
try to design efficient partitioning methods for it.
The main contributions of the paper can be summarized as follows.
\begin{itemize}
\item[\ding{172}] The formal definition of the 
balanced graph partition problem driven by optimizing workloads is introduced.
By transforming the problem to a representation of semidefinite program,
a bi-criteria  $O(\sqrt{\log n\log k})$-approximation algorithm with polynomial time
cost is proposed. 
\item[\ding{173}] The formal definition of the 
balanced graph partition problem based on optimizing motif computation is introduced.
The problem is proved to be $\lxmNP$-complete even for the special case of $k=2$ and the motif is a triangle.
Then, its inapproximability is also shown by proving that there are no efficient algorithms with finite approximation ratio.
\item[\ding{174}] For the special case that the motif is a triangle, using the semidefinite programming techniques,
a bi-criteria  $O(\sqrt{\log n\log k})$-approximation algorithm with polynomial time
cost is designed.
For the general case, the proposed algorithm is proved to be extended with the same performance guarantee. 
\end{itemize}

\section{Preliminaries}
In this part, we will first introduce some basic definitions and useful notations,
then, by comparing with the classical balanced graph partitioning problem,
the problems of balanced graph partitioning driven by optimizing workload and motif
are introduced.
Some important symbols utilized in the paper are summarized in Table~\ref{table:graph:notation}.

\begin{table}[t]
\centering
\caption{The List of Notations}\label{table:graph:notation}
	\vspace{0.5em}
	\begin{tabular}{|c|c|}
		\hline
		notation & comment \\
		\hline
		$G$ & the input graph \\
		$V_G$, $E_G$ & the node and edge set of $G$\\
		$k$ & the number of partitions\\
		$\textsc{PAR}_G$ & a partition solution of $G$\\
		$\textsc{PAR}^k_G$ & a $k$-partition solution of $G$\\
		$V_i$ & some partition \\
		$\tau$ & balance factor \\
		$\Phi$ & the given workload \\
		$M$ & the given motif \\
		$\textsf{cost}_{\Phi}(\cdot)$ & the cost function for optimizing workloads evaluation \\
		$\textsf{cost}_{M}(\cdot)$ & the cost function for optimizing motif computation \\
		\hline
	\end{tabular}
\end{table}

\subsection{Graph and Graph Partition}
In this paper, the input graph is denoted by $G=(V,E)$, where $V$ is the node set and $E\subseteq{V\times V}$ is the edge set.
When there are several graphs to be distinguished, we also use $V_G$ and $E_G$ to emphasize that the graph $G$ is referred.
\begin{definition}[the partition and $k$-partition solution of a graph]
Given a graph $G$, a partition solution is indeed defined on the node set $V$, that is $\textsc{PAR}_G=\{V_1,\cdots,V_m\}$.
Here, the union of all $V_i$s equals to $V$, and all $V_i$s are disjoint from each other.
If the size of the partition solution is $k$, that is $m=k$, it is also called to be a $k$-partition solution, denoted by $\textsc{PAR}^k_G=\{V_1,\cdots,V_k\}$.
\qed
\end{definition}

Moreover, for a node $u\in{V}$, we use $\textsc{PAR}_G(u)$ to represent the  partition where the node $u$ is placed.
Obviously, if $\textsc{PAR}_G(u)=V_i$, we always have $u\in{V_i}$.

With a clean context, we also use $\textsc{PAR}_G$ to represent $\textsc{PAR}^k_G$.
Specially, when $k$ is some constant, for example $k=2$, 2-partition is utilized for simplicity.

\subsection{The Classical Balanced Graph Partitioning Problem}
The classical balanced graph partitioning problem is indeed designed for load balancing,
whose formal definition is as follows.

\begin{definition}[$k$-balanced partition solution of a graph]\label{def:kBalancedPartitionA}
A $k$-balanced partitioning solution of $G$ is a $k$-partition solution $\textsc{PAR}_G=\{V_i\}$ such that for any $i\in[1,k]$ we have $|V_i|\leq{\lceil\frac{|V|}{k}\rceil}$.
\qed
\end{definition}

Since there are many possible partitioning solution, to select the best one,
the optimal partitioning solution depends on the partitioning cost function.
The most commonly used is defined based on the cut set, which can be expressed in Equation~\eqref{eq:costPartition}.
\begin{equation}
	\textsf{cost}(\textsc{PAR}_G)=|\{(u,v)\in{E}\;|\;\textsc{PAR}_G(u)\neq\textsc{PAR}_G(v)\}| \label{eq:costPartition}
\end{equation}

Then, to minimize the size of cutting edges, we have the following definition.
\begin{definition}[$k$-balanced graph partitioning, {\lxmProblemPrevious} for short]\label{def:kBalancedPartitionProblem}
Given a graph $G$ and an positive integer $k$, the $k$-balanced graph partitioning problem is to find 
a $k$-balanced partition solution $\textsc{PAR}_{G}$ for $G$ such that the cost $\textsf{cost}(\textsc{PAR}_G)$ is minimized.
\qed
\end{definition}

The {\lxmProblemPrevious} problem is one of the most classical $\lxmNP$-complete problems\cite{Garey79ComputersIntractablity}, which is also called the graph bisection problem
for the special case $k=2$.
\begin{theorem}[\cite{Garey79ComputersIntractablity}]\label{theorem:graph:kBGPNPC}
The {\lxmProblemPrevious} problem is $\lxmNP$-complete, even for the special case $k=2$.
\qed
\end{theorem}

Even worse, we have the following result.
\begin{theorem}[\cite{Andreev06BalancedGraph}]\label{theorem:graph:kBGPAppro}
For any $k\geq 3$, unless $\lxmPTime=\lxmNP$, there is no polynomial time algorithm for the {\lxmProblemPrevious} problem with finite approximation ratio. 
\qed
\end{theorem}

Due to the hardness shown by Theorem~\ref{theorem:graph:kBGPAppro}, more relaxations are considered and
the \emph{bi-criteria} approximation algorithms are designed. 
\begin{definition}[the $(k,\tau)$-balanced partition solution of a graph]\label{def:kBalancedPartitionB}
Given a balance factor $\tau$, a $(k,\tau)$-balanced partition solution of a graph $G$ is a $k$-partiton solution $\textsc{PAR}_G=\{V_i\}$,
such that for any $i\in[1,k]$ we have $|V_i|\leq{\tau\lceil\frac{|V|}{k}\rceil}$.
\qed
\end{definition}

Given a {\lxmProblemPrevious} instance $(G,k)$, the aim of a bi-criteria $(k,\tau)$ approximation algorithm 
$\mathcal{P}$ is to find a $(k,\tau)$-balanced partitioning solution $\textsc{PAR}$ for $G$ such that
the approximation ratio of $\mathcal{P}$ is defined as $\rho_{\mathcal{P}}=\frac{\textsf{cost}(\textsc{PAR})}{\textsf{cost}(\textsc{PAR}^*)}$
where $\textsc{PAR}^*$ is the optimal solution for the {\lxmProblemPrevious} problem.
It should be noted that, because of the definition of bi-criteria approximation algorithms, the
performance ratio $\rho_{\mathcal{P}}$ may be less than 1.

\subsection{Workload Driven Balanced Graph Partitioning Problem}
In big data computing platforms, there are many different kinds of computing tasks,
which may require different data and have different costs of computation, communication and so on.
Therefore, a key problem is to determine the partitioning solution to optimize specific workloads.
In this part, considering several kinds of typical computing tasks such as pattern matching,
traversal, shortest path and so on, the problem of balanced graph partitioning driven by workloads is introduced. 

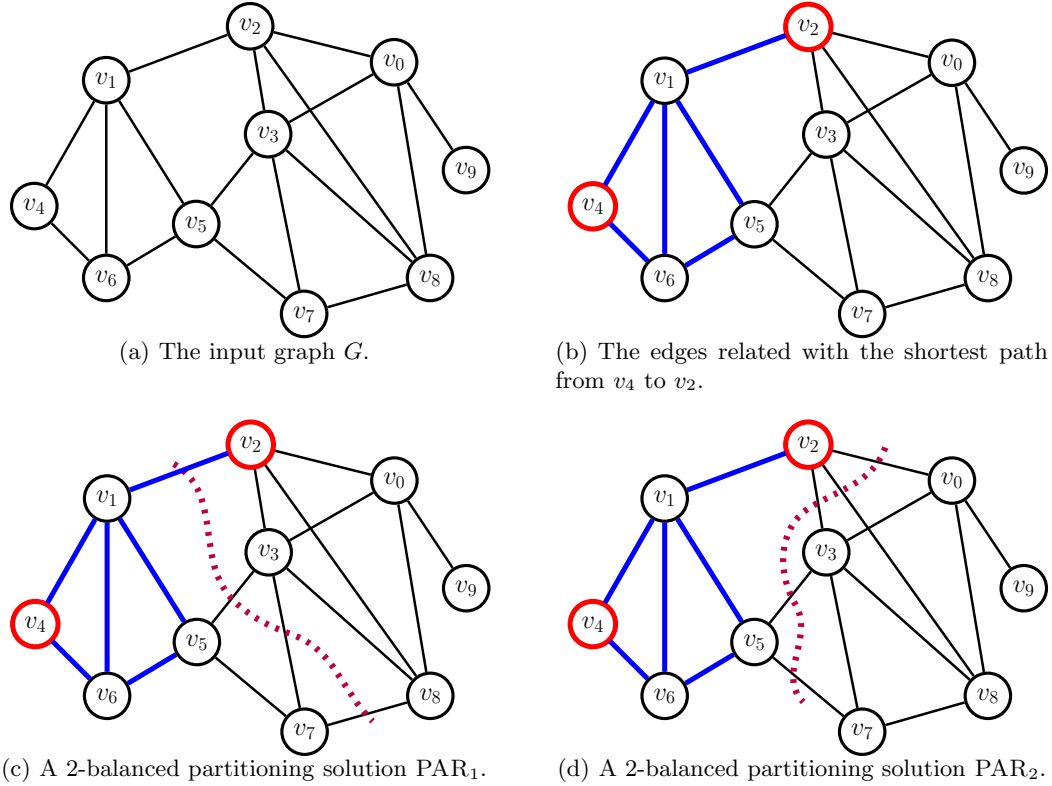
\begin{figure}[ht]
	\centering
	\subfigure[The input graph $G$.]{
		\label{fig:graph:workloadExample:A}
		\begin{minipage}{0.39\textwidth}\centering
			\resizebox{\textwidth}{!}{ %
				\begin{tikzpicture}[ auto,swap]
					\foreach \pos/\name/\vertexID in {{(0,2.5)/a/1}, {(4,4)/b/2}, {(4.5,1)/c/3},
						{(-2,-1)/d/4}, {(2.5,-1.5)/e/5}, {(0,-3)/f/6}, {(5.5,-4)/g/7}, {(9,-3)/h/8}, {(10,0)/i/9}, {(8,3)/j/0}}
					\node[vertex] (\name) at \pos {$v_{\vertexID}$};
					% Connect vertices with edges and draw weights
					\foreach \source/ \dest /\weight in {b/a/7, c/j/8,d/a/5,i/j/9,
						b/c/7, e/c/5,a/f/15,
						f/d/6,f/e/8,
						g/e/9,g/c/11,a/e/12,g/h/13,h/j/14,b/j/15,b/h/16,c/h/17}
					\path[edge] (\source) -- node[weight] {} (\dest);
			\end{tikzpicture}}
	\end{minipage}}
~~~~~~
	\subfigure[The edges related with the shortest path from $v_4$ to $v_2$.]{
		\label{fig:graph:workloadExample:B}
		\begin{minipage}{0.39\textwidth}\centering
			\resizebox{\textwidth}{!}{ %
				\begin{tikzpicture}[ auto,swap]
					\foreach \pos/\name/\vertexID in {{(0,2.5)/a/1}, {(4,4)/b/2}, {(4.5,1)/c/3},
						{(-2,-1)/d/4}, {(2.5,-1.5)/e/5}, {(0,-3)/f/6}, {(5.5,-4)/g/7}, {(9,-3)/h/8}, {(10,0)/i/9}, {(8,3)/j/0}}
					\node[vertex] (\name) at \pos {$v_{\vertexID}$};
					% Connect vertices with edges and draw weights
					\foreach \source/ \dest /\weight in {b/a/7, c/j/8,d/a/5,i/j/9,
						b/c/7, e/c/5,a/f/15,
						f/d/6,f/e/8,
						g/e/9,g/c/11,a/e/12,g/h/13,h/j/14,b/j/15,b/h/16,c/h/17}
					\path[edge] (\source) -- node[weight] {} (\dest);
					%\foreach \vertex / \fr in {d/1,a/2,b/3,g/6}
					\path node[selected vertex] at (d) {};
					\path node[selected vertex] at (b) {};
					\foreach \source / \dest in {d/a,d/f,a/f,a/b,a/e,f/e}
					\path[selected edge] (\source) -- (\dest);
			\end{tikzpicture}}
	\end{minipage}}\\
	\subfigure[A 2-balanced partitioning solution $\textsc{PAR}_1$.]{
		\label{fig:graph:workloadExample:C}
		\begin{minipage}{0.39\textwidth}\centering
			\resizebox{\textwidth}{!}{ %
				\begin{tikzpicture}[ auto,swap]
					\foreach \pos/\name/\vertexID in {{(0,2.5)/a/1}, {(4,4)/b/2}, {(4.5,1)/c/3},
						{(-2,-1)/d/4}, {(2.5,-1.5)/e/5}, {(0,-3)/f/6}, {(5.5,-4)/g/7}, {(9,-3)/h/8}, {(10,0)/i/9}, {(8,3)/j/0}}
					\node[vertex] (\name) at \pos {$v_{\vertexID}$};
					% Connect vertices with edges and draw weights
					\foreach \source/ \dest /\weight in {b/a/v1, c/j/v2,d/a/v3,i/j/v4,
						b/c/v5, e/c/v6,a/f/v7,
						f/d/v8,f/e/v9,
						g/e/v10,g/c/v11,a/e/v12,g/h/v13,h/j/v14,b/j/v15,b/h/v16,c/h/17}
					\path[edge] (\source) -- node[pos=0.5,minimum height=10pt,minimum width=0pt](\weight) {} (\dest);
					\draw[line width=5pt,loosely dashed,purple] (v1.center) to[out=-30,in=140] (v6.center) to[out=-40,in=160] (v11.center) to[out=-20] (v13.center);
					\path node[selected vertex] at (d) {};
					\path node[selected vertex] at (b) {};
					\foreach \source / \dest in {d/a,d/f,a/f,a/b,a/e,f/e}
					\path[selected edge] (\source) -- (\dest);
			\end{tikzpicture}}
	\end{minipage}}
~~~~~~
	\subfigure[A 2-balanced partitioning solution $\textsc{PAR}_2$.]{
		\label{fig:graph:workloadExample:D}
		\begin{minipage}{0.39\textwidth}\centering
			\resizebox{\textwidth}{!}{ %
				\begin{tikzpicture}[ auto,swap]
					\foreach \pos/\name/\vertexID in {{(0,2.5)/a/1}, {(4,4)/b/2}, {(4.5,1)/c/3},
						{(-2,-1)/d/4}, {(2.5,-1.5)/e/5}, {(0,-3)/f/6}, {(5.5,-4)/g/7}, {(9,-3)/h/8}, {(10,0)/i/9}, {(8,3)/j/0}}
					\node[vertex] (\name) at \pos {$v_{\vertexID}$};
					% Connect vertices with edges and draw weights
					\foreach \source/ \dest /\weight in {b/a/v1, c/j/v2,d/a/v3,i/j/v4,
						b/c/v5, e/c/v6,a/f/v7,
						f/d/v8,f/e/v9,
						e/g/v10,g/c/v11,a/e/v12,g/h/v13,h/j/v14,j/b/v15,b/h/v16,c/h/17}
					\path[edge] (\source) -- node[pos=0.5,minimum height=10pt,minimum width=0pt](\weight) {} (\dest);
					\draw[line width=5pt,loosely dashed,purple] (v15.north) to[out=-110,in=30] (v5.center) to[out=210,in=120] (v6.center) to[out=-60] (v10.south);
					\path node[selected vertex] at (d) {};
					\path node[selected vertex] at (b) {};
					\foreach \source / \dest in {d/a,d/f,a/f,a/b,a/e,f/e}
					\path[selected edge] (\source) -- (\dest);
			\end{tikzpicture}}
	\end{minipage}}
\caption{An Example of Balanced Graph Partitioning Driven by Workload}
	\label{fig:graph:workloadExample}
\end{figure}

\begin{example}\label{example:graph:workload}
Consider the graph $G$ shown in Fig.~\ref{fig:graph:workloadExample:A}, which is composed of 10 nodes.
Let $k=2$, that is we need to partition the nodes in $G$ into two parts.
Two possible partitioning solution are shown in Fig.~\ref{fig:graph:workloadExample:C} and \ref{fig:graph:workloadExample:D}, denoted by $\textsc{PAR}_1$ and $\textsc{PAR}_2$.
Here, the solution of $\textsc{PAR}_1$ is $\{v_1,v_4,v_5,v_6,v_7\}$ and $\{v_0,v_2,v_3,v_8,v_9\}$, whose cut edges are $\{(v_1,v_2),(v_3,v_5),(v_3,v_7),(v_7,v_8)\}$ and the cost is 4, that is $\textsf{cost}(\textsc{PAR}_1)=4$.
For the solution $\textsc{PAR}_2$, the partitioning result is $\{v_1,v_2,v_4,v_5,v_6\}$ and $\{v_0,v_3,v_7,v_8,v_9\}$, whose cut edges eare $\{(v_0,v_2),(v_2,v_8),(v_2,v_3),(v_3,v_5),(v_5,v_7)\}$ and the cost is 5, that is $\textsf{cost}(\textsc{PAR}_1)=5$.
According to Definition~\ref{def:kBalancedPartitionProblem}, $\textsc{PAR}_1$ is better than $\textsc{PAR}_2$.
However, if we consider some possible specific computing jobs, the result may be that $\textsc{PAR}_1$ is worse than $\textsc{PAR}_2$.
Assume that the shorted path between $v_4$ and $v_2$ needs to be computed on $G$, and it is assumed that the algorithm needs to do a breath first
travel beginning from $v_4$.
Then, it is easy to verify that the visited edges are the ones with blue color in Fig.~\ref{fig:graph:workloadExample:B}.
Comparing $\textsc{PAR}_1$ with $\textsc{PAR}_2$ again, it can be found that $\textsc{PAR}_2$ is better than $\textsc{PAR}_1$.
As shown in Fig.~\ref{fig:graph:workloadExample:C}, $\textsc{PAR}_1$ places the nodes of the edge $(v_1,v_2)$ into different partitions, which may cause
once communication cost, while as shown in Fig.~\ref{fig:graph:workloadExample:D}, $\textsc{PAR}_2$ does not cut any edges which will be visited by the
algorithm for shortest path, the communication cost will be 0.
\qed
\end{example}

As shown by Example~\ref{example:graph:workload}, the cost function utilized in Definition~\ref{def:kBalancedPartitionProblem}
is not proper for optimizing the specific workloads.
Next, based on a formal representation of workload, the problem of workload driven $k$-balanced graph partitioning, {\lxmProblemWordload} for short,
is introduced.

The workload on graph $G$ can be represented by a set $\Phi=\{(\phi_1,c_1),\cdots,(\phi_w,c_w)\}$, where each $\phi_i$
is a specific computing job in the workload, $c_i$ is a positive integer indicating the frequency that $\phi_i$ appears.
One job $\phi_i$ is a set  $\{(e_{ij},c_{ij})\}$, where $e_{ij}\in{E_G}$ is an edge of $G$ and $c_{ij}$ is a positive integer.
The element $(e_{ij},c_{ij})\in{\phi_i}$ means that the job $\phi_i$ needs to visit the edge $e_{ij}$ for $c_{ij}$ times, and if the two nodes
are placed on two different partitions, there may be $c_{ij}$ communications.

Given a partitioning solution $\textsc{PAR}$ and a workload representation $\Phi$ on $G$, let $\textsf{cost}_{\Phi}(\textsc{PAR})$
be the expected communication cost caused by partitioning $G$ according to $\textsc{PAR}$ and
computing $\Phi$ on $G$, which can be expressed by the following form.
\begin{align}
	\textsf{cost}_{\Phi}(\textsc{PAR}) &=\frac{1}{\sum_{1\leq i\leq w}c_i}\sum_{1\leq i\leq w}\big(c_i\cdot\textsf{cost}_{\phi_i}(\textsc{PAR})\big)\label{eq:graph:workloadCost}\\
	&=\frac{1}{\sum_{1\leq i\leq w}c_i}\sum_{1\leq i\leq w}\bigg(c_i\cdot\sum_{1\leq j\leq |\phi_i|}\big(c_{ij}\cdot\mathbf{I}(\textsc{PAR}(u_{ij})\neq\textsc{PAR}(v_{ij}))\big)\bigg)\notag
\end{align}
Here, $\textsf{cost}_{\phi_i}(\textsc{PAR})$ is the communication cost needed to evaluate the job $\phi_i$, the two nodes related with the edge $e_{ij}$ are $u_{ij}$ and $v_{ij}$,
and $\mathbf{I}(\cdot)$ is a function whose value is 1 if the required condition is satisfied, otherwise it is 0.

\begin{definition}[The {\lxmProblemWordload} Problem]\label{def:graph:WkBGP}
Given an input graph $G$, a positive integer $k$ and a workload set $\Phi$, the {\lxmProblemWordload} problem is to find a $k$-balanced partitioning solution
$\textsc{PAR}_G$ for $G$ such that the cost $\textsf{cost}_{\Phi}(\textsc{PAR}_G)$ of computing $\Phi$ is minimized.
\qed
\end{definition}

It is easy to verify that the {\lxmProblemPrevious} is a special case of the {\lxmProblemWordload} problem,
according to Theorem~\ref{theorem:graph:kBGPNPC}, we have the following result trivially.
\begin{proposition}
The {\lxmProblemWordload} problem is $\lxmNP$-complete. \qed
\end{proposition}

Then, the inapproximability result of {\lxmProblemPrevious} can be extended to the {\lxmProblemWordload} problem naturally.
\begin{proposition}
Unless $\lxmPTime=\lxmNP$, there is no polynomial time algorithm for {\lxmProblemWordload} with finite approximation ratio.
\qed
\end{proposition}

\subsection{Motif Based Balanced Graph Partitioning Problem}

Motif computation is a hot topic in the area of graph data processing, which has become
a typical kind of computing tasks in big data platforms.
The motivation of studying motif computation comes from the traditional graph matching problem,
which is to determine whether a given pattern graph $Q$ appears in a given data graph $G$.
Since the graph matching problem is intractable, the most popular and feasible
solution is to design efficient algorithms for the pattern graphs with small sizes.
Motif is just the graph of small sizes, such as 3 or 4 nodes.
In this part, the motif based balanced graph partitioning problem ({\lxmProblemMotif} for short) is introduced.

\begin{figure}[ht]
	\centering
	\subfigure[A 2-balanced partitioning solution $\textsc{PAR}_1$ of $G$ and the number of motifs destroyed by $\textsc{PAR}_1$]{
		\label{fig:graph:motifExample:A}
		\begin{minipage}{0.39\textwidth}\centering
			\resizebox{\textwidth}{!}{ %
				\begin{tikzpicture}[ auto,swap]
					\foreach \pos/\name/\vertexID in {{(0,2.5)/a/1}, {(4,4)/b/2}, {(4.5,1)/c/3},
						{(-2,-1)/d/4}, {(2.5,-1.5)/e/5}, {(0,-3)/f/6}, {(5.5,-4)/g/7}, {(9,-3)/h/8}, {(10,0)/i/9}, {(8,3)/j/0}}
					\node[vertex] (\name) at \pos {$v_{\vertexID}$};
					% Connect vertices with edges and draw weights
					\foreach \source/ \dest /\weight in {b/a/v1, c/j/v2,d/a/v3,i/j/v4,
						b/c/v5, e/c/v6,a/f/v7,
						f/d/v8,f/e/v9,
						g/e/v10,g/c/v11,a/e/v12,g/h/v13,h/j/v14,b/j/v15,b/h/v16,c/h/17}
					\path[edge] (\source) -- node[pos=0.5,minimum height=10pt,minimum width=0pt](\weight) {} (\dest);
					\draw[line width=5pt,loosely dashed,purple] (v1.center) to[out=-30,in=140] (v6.center) to[out=-40,in=160] (v11.center) to[out=-20] (v13.center);
					\foreach \source / \dest in {c/e,c/g,g/e,g/h,c/h}
					\path[selected edge] (\source) -- (\dest);
			\end{tikzpicture}}
	\end{minipage}}
	~~~~~~
	\subfigure[A 2-balanced partitioning solution $\textsc{PAR}_3$ of $G$ and the number of motifs destroyed by $\textsc{PAR}_3$]{
		\label{fig:graph:motifExample:B}
		\begin{minipage}{0.39\textwidth}\centering
			\resizebox{\textwidth}{!}{ %
				\begin{tikzpicture}[ auto,swap]
					\foreach \pos/\name/\vertexID in {{(0,2.5)/a/1}, {(4,4)/b/2}, {(4.5,1)/c/3},
						{(-2,-1)/d/4}, {(2.5,-1.5)/e/5}, {(0,-3)/f/6}, {(5.5,-4)/g/7}, {(9,-3)/h/8}, {(10,0)/i/9}, {(8,3)/j/0}}
					\node[vertex] (\name) at \pos {$v_{\vertexID}$};
					% Connect vertices with edges and draw weights
					\foreach \source/ \dest /\weight in {b/a/v1, c/j/v2,d/a/v3,i/j/v4,
						b/c/v5, e/c/v6,a/f/v7,
						f/d/v8,f/e/v9,
						e/g/v10,g/c/v11,a/e/v12,g/h/v13,h/j/v14,j/b/v15,b/h/v16,c/h/17}
					\path[edge] (\source) -- node[pos=0.5,minimum height=10pt,minimum width=0pt](\weight) {} (\dest);
					\draw[line width=5pt,loosely dashed,purple] (v1.center) to[out=-30,in=140] (v6.center) to[out=-60] (v10.south);
					\draw[line width=5pt, dashed,purple] (v4.north) to ([xshift=-.3cm,yshift=-.8cm]v4.center);
					\foreach \source / \dest in {c/e,c/g,g/e}
					\path[selected edge] (\source) -- (\dest);
			\end{tikzpicture}}
	\end{minipage}}
\caption{An Example of Balanced Graph Partitioning Driven by Motif Computation}
	\label{fig:graph:motifExample}
\end{figure}
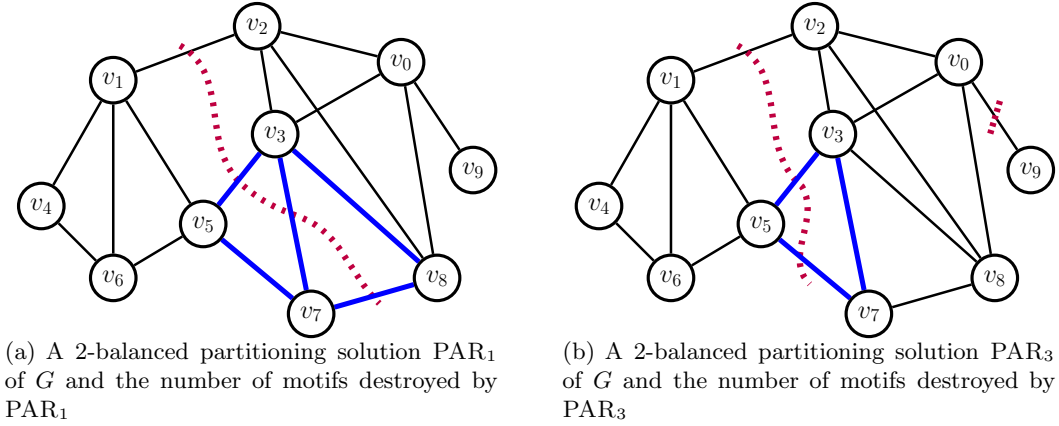

\begin{example}\label{example:graph:motif}
Consider the graph $G$ shown in Fig.~\ref{fig:graph:workloadExample:A}, and let $k=2$, that is we need to compute a 2-balanced partitioning solution.
There are two possible 2-balanced partitioning solutions shown in Fig.~\ref{fig:graph:motifExample:A} and \ref{fig:graph:motifExample:B}.
The partitioning solution shown in Fig.~\ref{fig:graph:motifExample:A} is the same as the one shown in Fig.~\ref{fig:graph:workloadExample:C},
we still denote it by $\textsc{PAR}_1$.
The solution shown in Fig.~\ref{fig:graph:motifExample:B} is denoted by $\textsc{PAR}_3$.
In $\textsc{PAR}_1$, the graph $G$ is partitioned into two sets $\{v_1,v_4,v_5,v_6,v_7\}$ and $\{v_0,v_2,v_3,v_8,v_9\}$, the corresponding cut edges are $\{(v_1,v_2),(v_3,v_5),(v_3,v_7),(v_7,v_8)\}$, 
and the partitioning cost is 4, that is $\textsf{cost}(\textsc{PAR}_1)=4$.
In $\textsc{PAR}_3$, the graph $G$ is partitioned into two sets $\{v_1,v_4,v_5,v_6,v_9\}$ and $\{v_0,v_2,v_3,v_7,v_8\}$,
the corresponding cut edges are $\{(v_1,v_2),(v_3,v_5),(v_5,v_7),(v_0,v_9)\}$,
and the partitioning cost is also 4, that is $\textsf{cost}(\textsc{PAR}_3)=4$.
It can be found that, using the traditional partitioning cost measure, no one of the two solutions is better than the other one.
Next, let us consider a specific job of motif computation, whose aim is to calculate the number of triangles in $G$.
Obviously, no matter which method is utilized, if the three nodes of one triangle are not in one partition, there
will be at least one communication.
In Fig.~\ref{fig:graph:motifExample:A} and \ref{fig:graph:motifExample:B}, the triangles destroyed by the partitioning solution are
colored with blue.
Then, by counting the number of triangles destroyed, the solution $\textsc{PAR}_3$ should be better than $\textsc{PAR}_1$.
\qed
\end{example}

According to Example~\ref{example:graph:motif}, it can be found that the $\textsf{cost}(\cdot)$ function used by Definition~\ref{def:kBalancedPartitionProblem} is
based on counting the cut edges and it is not proper for optimizing the tasks of motif computation.
Next, the cost function for motif computation is introduced first, and then the definition
of the motif computation based $k$-balanced graph partitioning ({\lxmProblemMotif} for short) is introduced.

Let $c$ be a constant, a $c$-Motif is a graph with $c$ nodes.
Given a partitioning solution  $\textsc{PAR}$ of $G$ and a motif $M$, intuitively, let $\textsf{cost}_{M}(\textsc{PAR})$ be the communication cost of processing
computation of $M$ on partitions obtained according to $\textsc{PAR}$.
\begin{align}
	\textsf{cost}_{M}(\textsc{PAR}) &=\sum_{H\in{G\langle M\rangle}}\mathbf{I}\big(\exists(u,v)\in{H}\text{s.t.}\textsc{PAR}(u)\neq\textsc{PAR}(v)\big)\label{eq:graph:motifCost}
\end{align}
Here, $G\langle M\rangle$ means all appearance of the motif $M$ in the graph $G$, and $\mathbf{I}(\cdot)$ is the indicator function whose value is in  \{0,1\}.

\begin{definition}[The {\lxmProblemMotif} problem]\label{def:graph:MkBGP}
Given a graph $G$, a positive interger $k$ and a $c$-Motif $M$,
the {\lxmProblemMotif} problem is to compute a $k$-balanced partitioning solution $\textsc{PAR}_{G}$ for $G$
such that the cost $\textsf{cost}_{M}(\textsc{PAR}_G)$ is minimal.
\qed
\end{definition}

Different from the {\lxmProblemWordload} problem, the {\lxmProblemPrevious} problem is not a special case of {\lxmProblemMotif},
therefore, the complexity result of {\lxmProblemMotif} can not be obtained directly from Theorem~\ref{theorem:graph:kBGPNPC}.

\section{Graph Partitioning Algorithm for Workload Optimization}
\label{sec:graph:wkbgpalg}

\subsection{Semidefinite Programming}
\eat{%%%%%%%%%%%%%%%%%%%
Semidefinite programming\cite{Alizadeh95InteriorPoint} is a sophisticated technique for solving complex optimization problem,
which can be solved within polynomial time in practice \cite{Alizadeh95InteriorPoint}.
}%%%%%%%%%%%%%%%%%%%%%%%%

Let $\bm{X}$ be a symmetric square $n\times n$ matrix. 
If all eigenvalues of $\bm{X}$ are non positive, it is also called a positive semidefinite matrix, psd for short.
\eat{%%%%%%%%%%%%%%%%%
Then, the following statements can be proved to be equivalent.
(1) $\bm{X}$ is psd, denoted by $\bm{X}\succcurlyeq 0$.
(2) For any $n$-dimensional vector $\bm{y}\in\mathbb{R}^n$, we have  $\bm{y}^T\bm{X}\bm{y}\geq 0$.
(3) There exists a $m\times n$ matrix $\bm{V}$ such that $\bm{X}=\bm{V}^T\bm{V}$.
}%%%%%%%%%%%%%%%%%%%

\eat{%%%%%%%%%%%%%%%%%%%%%%%%%%%%%%%%%%%
Accoridng to \cite{Du11ApproximationAlgorithm}, given two square matrixs $\bm{A}=(a_{ij})_{n\times n}$ and $\bm{B}=(b_{ij})_{n\times n}$,
the \emph{Frobenius product} of $\bm{A}$ and $\bm{B}$ can be defined as follows.
\begin{equation}\label{eq:graph:frobeniusproduct}
	\bm{A}\bullet\bm{B}=\sum_{i=1}^{n}\sum_{j=1}^{n}a_{ij}b_{ij}
\end{equation}

Then, a general semidefinite program, SDP for short, can be expressed as follows, where $\bm{X}=(x_{ij})_{n\times n}$ is a matrix of variables, $\bm{C}=(c_{ij})_{n\times n}$ and $\bm{A}_k=(a^k_{ij})_{n\times n}$ ($k\in[1,m]$) are used
to express the constraints. 
\begin{align}\label{eq:graph:sdpdefinition}
	\text{minimize} \hspace{10pt}& \bm{C}\bullet\bm{X}\\
	\text{s.t.} \hspace{10pt}& \bm{A}_k\bullet\bm{X}=b_k,\hspace{20pt}\forall k\in[1,l]\notag\\
	%& \bm{C}=(c_{ij})_{n\times n}, \bm{A}_k=(c^k_{ij})_{n\times n}\notag\\
	%& \bm{X}=(x_{ij})_{n\times n}, \bm{X}\succcurlyeq 0\notag
	& \bm{X}\succcurlyeq 0\notag
\end{align}
}%%%%%%%%%%%%%%%%%%%%%%%%%%%%%%%%%%%

%%%%%%%%%%%%%%%%%%%%%%%%%%%%%%%%%%%%%%%%%%%%%%%%%%%%%%%%%%%%%%

According to the characteristics of semidefinite matrix, the matrix $\bm{X}$ can be represented by $\bm{X}=\bm{V}^T\bm{V}$ also, where $\bm{V}$ is a $m\times n$ matrix.
%As shown by Fig.~\ref{fig:graph:sdpToVectorExample}, 
If each column of the matrix $\bm{V}$ is treated as a $m$-dimensional vector,
that is $\bm{V}$ is composed of $n$ vectors $\bm{v}_1$, $\bm{v}_2$, $\ldots$, $\bm{v}_n$.
Then, each element $x_{ij}$ in $\bm{X}$ can be represented by $\sum_{1\leq h\leq m}v_{ih}v_{hj}$.
Obviously, $x_{ij}$ can be represented by $\bm{v}^T_i\bm{v}_j$ or $\bm{v}_i\cdot\bm{v}_j$.
%Therefore, the semidefinite program show in Equation~\eqref{eq:graph:sdpdefinition} can be also rewritten as the following form.
Therefore, a typical semidefinite program can be represented by the following form.
\begin{align}\label{eq:graph:sdpVecdefinition}
	\text{minimize} \hspace{10pt}& \sum_{i,j}c_{ij}(\bm{v}_i\cdot\bm{v}_j)\\
	\text{s.t.} \hspace{10pt}& \sum_{i,j}a^k_{ij}(\bm{v}_i\cdot\bm{v}_j)
	=b_k,\hspace{20pt}\forall k\in[1,l]\notag\\
	%& \bm{C}=(c_{ij})_{n\times n}, \bm{A}_k=(c^k_{ij})_{n\times n}\notag\\
	%& \bm{X}=(x_{ij})_{n\times n}, \bm{X}\succcurlyeq 0\notag
	& \bm{v}_i\in\mathbb{R}^m,\hspace{20pt}\forall i\in[1,n]\notag
\end{align}

According to \cite{Alizadeh95InteriorPoint,Du11ApproximationAlgorithm}, 
given a SDP of size $n$ and $\epsilon>0$, there exists a $O\big((n+1/\epsilon)^{O(1)}\big)$-time algorithm solving the SDP with an approximation ratio $(1+\epsilon)$.

\subsection{The Semidefinite Program for {\lxmProblemWordload}}
In this part, we will represent the {\lxmProblemWordload} problem using semidefinite program.
The intuitive idea can be explained as follows.
Each node in $V$ can be represented as a $m$-dimensional vector, where $m\geq k$, and it is denoted by $\bm{v}$.
Ideally, it is expected that the vector $\bm{v}$ can indicate which partition the node should be placed into.
When  two nodes are in the same partition, the corresponding vectors should be same or very similar.
Given the input $(G,k,\Phi)$ of the {\lxmProblemWordload} problem, the semidefinite representation can be expressed as Equation~\eqref{eq:graph:WkBGPsdp:A}.
\begin{align}
	%\min\hspace{1em} \frac{1}{\sum_{1\leq i\leq w}c_i}\bigg(c_i\cdot\sum_{1\leq j\leq |\phi_i|}\big(c_{ij}\cdot\mathbf{I}(\textsc{PAR}(u_{ij})\neq\textsc{PAR}(v_{ij}))\big)\bigg)\\
	\min\hspace{1em} &\frac{1}{\sum_{1\leq i\leq w}c_i}\sum_{1\leq i\leq w}\bigg(c_i\cdot\sum_{1\leq j\leq |\phi_i|}\big(c_{ij}\cdot\frac{1}{2}\|\bm{u}_{ij}-\bm{v}_{ij}\|^2_2\big)\bigg)\label{eq:graph:WkBGPsdp:A}\\
	\text{s.t.}\hspace{1em}&\|\bm{u}-\bm{v}\|_2^2+\|\bm{v}-\bm{w}\|_2^2\geq \|\bm{u}-\bm{w}\|_2^2\hspace{2em}\forall u,v,w\in{V}\label{eq:graph:WkBGPsdp:B}\\
	&\sum_{v\in{S}}\frac{1}{2}\|\bm{u}-\bm{v}\|^2_2\geq|S|-\frac{n}{k}\hspace{5em}\forall S\subseteq{V},u\in{S}\label{eq:graph:WkBGPsdp:C}
\end{align}

Here, $\|\cdot\|_2^2$ used in Equation~\eqref{eq:graph:WkBGPsdp:A} is the square of $L2$ normal form, that is $\ell_2^2$, which can be
defined as $\|\bm{x}\|^2_2=\sum_{i=1}^{m}x_i^2$.
The first constraint shown by Equation~\eqref{eq:graph:WkBGPsdp:B} is the triangle inequality defined on $\ell_2^2$, which
is firstly proposed by \cite{Arora04ExpanderFlows} and used to express the partitioning requirements.
The second constraint shown by Equation~\eqref{eq:graph:WkBGPsdp:C} is also called the spreading constraint, which
is proposed by \cite{Even00DivideConquer} and used to control the size of each partition.
The optimizing goal shown in Equation~\eqref{eq:graph:WkBGPsdp:A} is a variant of the cost function $\textsf{cost}_{\Phi}(\cdot)$ utilized in the
definition of the {\lxmProblemWordload} problem, where the part $\mathbf{I}(\textsc{PAR}(u_{ij})\neq\textsc{PAR}(v_{ij}))$ of $\textsf{cost}_{\Phi}(\cdot)$
is replaced with $\frac{1}{2}\|\bm{u}_{ij}-\bm{v}_{ij}\|^2_2$.
Intuitively, for two nodes $u_{ij}$ and $v_{ij}$, if the distance between $\bm{u}_{ij}$ and $\bm{v}_{ij}$ is too large (like $\approx 2$), it means that
$\textsc{PAR}(u_{ij})\neq\textsc{PAR}(v_{ij})$, otherwise, it means that $\textsc{PAR}(u_{ij})=\textsc{PAR}(v_{ij})$.
Therefore, the optimizing goal shown by Equation~\eqref{eq:graph:WkBGPsdp:A} is indeed a relaxed form of the cost function
$\textsf{cost}_{\Phi}(\cdot)$, and the transformation is roughly equivalent.

It should be noted that, given an instance of the {\lxmProblemWordload} problem, the semidefinite program shown in Equation~\eqref{eq:graph:WkBGPsdp:A}
will not be constructed explicitly, because the size of the constraints is huge. 
More specifically, the goal of Equation~\eqref{eq:graph:WkBGPsdp:A} can be built within $O(|\Phi|)$,
and the constraints in Equation~\eqref{eq:graph:WkBGPsdp:B} are built with respect to each triple of nodes in $G$ and it can be
finished within $O(|V_G|^3)$ obviously.
However, the constraints shown in Equation~\eqref{eq:graph:WkBGPsdp:C} are built based on each possible subset of $V_G$,
therefore, the size of the constraints can reach $\Theta(|V_G|\cdot 2^{|V_G|})$, which can not be built within polynomial time.
A feasible way to process the above constraints is to built a specific constraint when it is needed,
since the constraint can be defined over every possible subset of $V_G$ and it can be checked efficiently.
Therefore, the semidefinite program is solved based on the optimizing goal and the constraints in Equation~\eqref{eq:graph:WkBGPsdp:B} first,
and then for each $v\in{V_{G}}$ the corresponding constraint can be checked as the following way.
First, all nodes like $u\in V\setminus\{v\}$ are computed, and then the corresponding distance $\|\bm{u}-\bm{v}\|_2^2$ is calculated and the nodes
will be sorted into an increasing order. Finally, the nodes will be added into the set $S$ and checked whether the condition defined by Equation~\eqref{eq:graph:WkBGPsdp:C}
is satisfied.
According to Lemma~\ref{lemma:graph:sdp}, we have the following result.

\begin{theorem}\label{theorem:graph:sdpSolution}
The semidefinite program shown in Equation~\eqref{eq:graph:WkBGPsdp:A} can be solved in polynomial time.\qed
\end{theorem}
\begin{proof}
Let $N$ be the size of the input instance of the {\lxmProblemWordload} problem,
where $n$ is the number of nodes in $G$.
Obviously, we have $N=\max\{n,|\Phi|\}$.
Let us consider the range value of the optimized goal of Equation~\eqref{eq:graph:WkBGPsdp:A}.
First, the part related with $c_i$ can be treated as the weighted sum of the expression $\sum_{1\leq j\leq |\phi_i|}(c_{ij}\cdot\frac{1}{2}\|\bm{u}_{ij}-\bm{v}_{ij}\|_2^2)$,
and it can not be larger than $\max\{\sum_{1\leq j\leq |\phi_i|}(c_{ij}\cdot\frac{1}{2}\|\bm{u}_{ij}-\bm{v}_{ij}\|_2^2)\;|\;i\in[1,w]\}$.
Let $\gamma$ be the constant $\max\{\sum_{1\leq j\leq |\phi_i|}c_{ij}\;|\;i\in[1,w]\}$.
If the optimizing goal of Equation~\eqref{eq:graph:WkBGPsdp:A} is divided by $\gamma$, which will not change the optimized partitioning solution obviously,
the optimizing goal after the division will not be larger than the weighted sum of the  distance $\frac{1}{2}\|\bm{u}_{ij}-\bm{v}_{ij}\|_2^2$ over
all edges $(u_{ij},v_{ij})\in{E_G}$.
Then, let the edge with the maximized distance be $(u_{max},v_{max})$ and the corresponding distance be $\frac{1}{2}\|\bm{u}_{max}-\bm{v}_{max}\|_2^2$.
As a result, the final goal can be limited within the range $[0,|E_G|\cdot \frac{1}{2}\|\bm{u}_{max}-\bm{v}_{max}\|_2^2]$.

Moreover, let us consider a trivial solution for the semidefinite program.
First, randomly select $\frac{n}{k}$ nodes of $V$ and represent them as $v_1,\cdots,v_{n/k}$.
The corresponding $k$-dimensional vectors is set to be $\bm{v}_i=[\sqrt{n},0,0,\cdots,0]^T$.
Then, repeat the above operations until all nodes in $V$ have been selected, for each step,
the value $\sqrt{n}$ is moved backward one position.
Since the solution should be $m$-dimensional vectors, for the other $m-k$ dimensions, let the corresponding value be 0 always.

It is not hard to verify that the above trivial solution satisfy the two groups of constraints.
First, for two given nodes $u$ and $v$, the corresponding vector distance $\frac{1}{2}\|\bm{u}-\bm{v}\|^2_2$ only has two possible values, 0 or $n$, and
the value is 0 if and only if the two nodes are selected in the same step during the above construction.
Then, for the constraints shown in Equation~\eqref{eq:graph:WkBGPsdp:B}, it can be verified by analyzing the distances among three nodes,
which have the following two cases.
(1) The three nodes have same vectors, then the values of $\|\bm{u}-\bm{v}\|_2^2$, $\|\bm{u}-\bm{w}\|_2^2$ and $\|\bm{w}-\bm{v}\|_2^2$ are all 0, the constraints are satisfied trivially.
(2) There are at least two nodes having different vectors. Without loss of generality, we assume that $\bm{u}$ and $\bm{v}$ are different, then, at least two of the three values $\|\bm{u}-\bm{v}\|_2^2$, $\|\bm{u}-\bm{w}\|_2^2$ and $\|\bm{w}-\bm{v}\|_2^2$ are $2n$. Then, the constraints can be satisfied too.
Second, to verify the constraints shown in Equation~\eqref{eq:graph:WkBGPsdp:C} are satisfied, let us consider two cases based on the size of $S$.
(1) The first case is $|S|\leq\frac{n}{k}$. For this case, the right side of the constraint satisfies $|S|-\frac{n}{k}\leq 0$ and the left side is at least 0.
Therefore, the constraints can be satisfied.
(2) The second case is $|S|>\frac{n}{k}$.
For this case, since during the construction of the trivial solution each step only selects at most $n/k$ nodes,
in $S$, there must be at least one node $v$ whose vector is different from the one of $u$.
Then, for the left side of the constraint, we have $\sum_{v\in{S}}\frac{1}{2}\|\bm{u}-\bm{v}\|_2^2\geq n$, and for the right side, there must be $|S|-\frac{n}{k}\leq n-\frac{n}{k}<n$.
Thus, the condition $\sum_{v\in{S}}\frac{1}{2}\|\bm{u}-\bm{v}\|_2^2\geq|S|-\frac{n}{k}$ can be satisfied obviously.

Because the trivial solution constructed is a feasible solution, the corresponding cost	is an upper bound of the optimize cost,
then the value of the optimizing goal can be set within the range $[0,|V_G|\cdot|E_G|]$.
Further, according to Lemma~\ref{lemma:graph:sdp}, by setting the parameter $\epsilon$ to satisfy $\epsilon<\frac{1}{|V_G|\cdot|E_G|}$,
the $(1+\epsilon)$-approximation solution of the corresponding semidefinite program is an exact solution in fact.
Finally, the time cost of can be bounded by a polynomial of $(n+|V_G|\cdot|E_G|)$.
The proof is finished.
\qed
\end{proof}
\subsection{SDP Based Approximation Algorithm for WkBGP}
This part introduces the partitioning algorithm solving the {\lxmProblemWordload} problem,
which computes an approximation partitioning solution for the given {\lxmProblemWordload} instance
using the semidefinite programming and a well designed rounding techniques.
The performance guarantee of the approximation solution can be bounded based on analyzing the effects
of the rounding method.

\begin{algorithm}[t]
	\caption{{\lxmAlgorithmWorkload}~(Workload driven $k$ Balanced Graph Partitioning)\label{alg:graph:WkBGPAlg}}
	\SetKwProg{Fn}{function}{}{end}
%	\SetKwInput{KwIn}{{\normalfont\lxmFontFangSong 输入}}
%	\SetKwInput{KwOut}{{\normalfont\lxmFontFangSong 输出}}
	\KwIn{A Graph $G=\{V,E\}$, an integer $k$ and a workload set $\Phi$.}
	\KwOut{A partition $\textsc{PAR}_G=\{V_1,V_2,\cdots,V_k\}$ for $G$.}
	\Fn{\textsc{WkBGPartition}($G, k, \Phi$)}{%
		{Initialize a $|V_G|\times|V_G|$ array $W$ to be 0}\;
		{Let $C=\sum_{1\leq i\leq w}c_i$}\;
		\For{each $(\phi_x,c_x)\in{\Phi}$}{
			\For{each $(e_{xy},c_{xy})\in{\phi_x}$}{
				{Let the two nodes of $e_{xy}$ be $v_i$ and $v_j$}\;
				{$W[i][j]+=\frac{c_x\cdot c_{xy}}{C}$}\;
				{$W[j][i]=W[i][j]$}\;
			}
		}
		{Build sdp \eqref{eq:graph:WkBGPsdp:Improve}  using $\sum_{(v_i,v_j)\in{E_G}}W[i][j]\cdot\frac{1}{2}\|\bm{v}_i-\bm{v}_j\|_2^2$ as the goal}\;
		{Let $opt=\{\bm{v}_1,\cdots,\bm{v}_n\}$ be the optimal solution of sdp \eqref{eq:graph:WkBGPsdp:Improve}}\;
		{Invoke $k$-\textsc{Partition} of \cite{Krauthgamer09PartitioningGraphs} on $opt$ to compute $\textsc{PAR}_G=\{V_1,\cdots,V_k\}$}\;
		%		{$V'\leftarrow V$}\;
		%		\While{$|V'|>(1+2\epsilon)\frac{n}{k}$}{
			%			{Choose $\bm{r}\in\mathbb{R}^m$ such that $r_i$ is chosen \emph{i.i.d.} according to $N(0,1)$}\;
			%			{$tmp\leftarrow\emptyset$}\;
			%			\For{each $v\in{V'}$}{
				%				{a}\;
				%			}
			%		}
		\textbf{return} {$\textsc{PAR}_G$}\;
	}
\end{algorithm}

Given the workload input $\Phi$, different jobs may use same edges, obviously, in a partitioning solution,
it should be determined whether two nodes are placed in the same partition.
To use the semidefinite programming to solve the {\lxmProblemWordload} problem,
by combining the variables related with the same edge, we can rewritten the problem by the
following Equation~\eqref{eq:graph:WkBGPsdp:Improve}, which can be verified to be equivalent to Equation~\eqref{eq:graph:WkBGPsdp:A}.

\begin{align}
	%\min\hspace{1em} \frac{1}{\sum_{1\leq i\leq w}c_i}\bigg(c_i\cdot\sum_{1\leq j\leq |\phi_i|}\big(c_{ij}\cdot\mathbf{I}(\textsc{PAR}(u_{ij})\neq\textsc{PAR}(v_{ij}))\big)\bigg)\\
	\min\hspace{1em} &\sum_{(v_i,v_j)\in{E_G}}\big(W_{ij}\cdot\frac{1}{2}\|\bm{v}_i-\bm{v}_j\|_2^2\big)\label{eq:graph:WkBGPsdp:Improve}\\
	\text{s.t.}\hspace{1em}&\|\bm{u}-\bm{v}\|_2^2+\|\bm{v}-\bm{w}\|_2^2\geq \|\bm{u}-\bm{w}\|_2^2\hspace{2em}\forall u,v,w\in{V}\notag\\
	&\sum_{v\in{S}}\frac{1}{2}\|\bm{u}-\bm{v}\|^2_2\geq|S|-\frac{n}{k}\hspace{5em}\forall S\subseteq{V},u\in{S}\notag\\
	&W_{ij}=\sum_{1\leq x\leq w}\sum_{1\leq y\leq |\phi_x|}\frac{c_x\cdot c_{xy}\cdot\mathbf{I}(e_{xy}=(v_i,v_j))}{C} \notag\\
	& C=\sum_{1\leq i\leq w} c_i\notag
\end{align}

Next, the {\lxmAlgorithmWorkload} algorithm for solving the {\lxmProblemWordload} problem is introduced,
whose details can be found in Algorithm~\ref{alg:graph:WkBGPAlg}.
The main procedure of {\lxmAlgorithmWorkload} is to first construct the semidefinite program according to Equation~\eqref{eq:graph:WkBGPsdp:Improve}
based on the input graph $G$, the workload  $\Phi$ and the integer $k$, then use the program solver to obtain the optimized solution for
the relaxed variant, finally, computes the result partitioning solution by rounding the vectors obtained in the optimized solution.
The transformation between the two semidefinite programs is done under the help of a $|V_G|\times|V_G|$ array $W$.
Specifically, {\lxmAlgorithmWorkload} scans all jobs maintained in the workload set $\Phi$, and 
adds the corresponding weight related with each edge used in the computing jobs into $W$ (line 2-8).
Then, {\lxmAlgorithmWorkload} utilizes $W$ to construct the semidefinite program shown in Equation~\eqref{eq:graph:WkBGPsdp:Improve} (line 9).
Let $opt=\{\bm{v}_1,\cdots,\bm{v}_n\}$ be the optimized solution for the obtained program, according to Theorem~\ref{theorem:graph:sdpSolution}, $opt$ can be
computed within polynomial time (line 10).
Finally, using the rounding techniques proposed in \cite{Krauthgamer09PartitioningGraphs} (line 11), the final partitioning solution can be obtained. 
In details, \cite{Krauthgamer09PartitioningGraphs} proposes a rounding method named $k$-\textsc{Partition},
which can output a partitioning solution based on $opt$ as follows.
First, $k$-\textsc{Partition} utilizes the space embedding techniques proposed by \cite{Arora09ExpanderFlows}
to map the optimized solution from $\ell_2^2$ to $\ell_2$, which will enlarge the solution with a ratio at most $O(\sqrt{\log n})$.
Then, using the random hyperplane rounding techniques (see \cite{Du11ApproximationAlgorithm}),
the partitions with size at most $(1+2\epsilon)\frac{n}{k}$ can be built and extracted one by one, until all nodes have been processed once.

Next, we will prove the performance guarantee of the approximation ratio of {\lxmAlgorithmWorkload}.
\begin{theorem}\label{theorem:graph:WkBGPLowerBound}
Given the input $G$, $k$ and $\Phi$ of the {\lxmProblemWordload} problem,
the cost of the optimal solution of the Equation~\eqref{eq:graph:WkBGPsdp:A} (or \eqref{eq:graph:WkBGPsdp:Improve})
is a lower bound of the partitioning cost of the optimal solution of {\lxmProblemWordload}.
\qed
\end{theorem}
\begin{proof}
Let $\textsc{PAR}_G^*$ be the optimal partitioning solution of the given instance of {\lxmProblemWordload},
and let the corresponding partitioning cost be $\textsf{cost}_{\Phi}(\textsc{PAR}_G^*)$.
Then, assume that the optimal solution of the corresponding semidefinite program is
$\{\bm{u}_1,\bm{u}_2,\cdots,\bm{u}_n\}$ and the corresponding optimal value of the goal is $A^*$.
We can build a feasible solution $\{\bm{v}_1,\bm{v}_2,\cdots,\bm{v}_n\}$ of the program based on $\textsc{PAR}_G^*$,
whose details can be explained as follows.
Let $\textsc{PAR}_G^*$ be the corresponding $k$-balanced partitioning solution $\{V_1,V_2,\cdots,V_k\}$,
where for each $V_i$ we have $|V_i|=\frac{n}{k}$.
Then, for any node $v_i\in{V_j}$, let $v_i$ be the following vector.
	\begin{align}
		\bm{v}_i=[0,&\cdots,1,\cdots,0]^T\\
		&\hspace{.66em}\stackrel{\big\uparrow}{\text{the }j\text{th bit}}\notag
	\end{align}
	
Next, checking the constraints in Equation~\eqref{eq:graph:WkBGPsdp:A}, we can verify the correctness
of the solution obtained.

(1) For the constraints in Equation~\eqref{eq:graph:WkBGPsdp:B},
since the $\|\cdot\|_2^2$ distance between two nodes must be 0 or 2,
the only possible way to violate the constraints is that both $\|\bm{u}-\bm{v}\|_2^2$ and $\|\bm{v}-\bm{w}\|_2^2$ are 0 but $\|\bm{u}-\bm{w}\|_2^2$ is 2.
For this case, the nodes $u$ and $v$ are placed into a same partition, $u$ and $v$ belong to 
one partition, and $u$ and $w$ are placed into two different partitions.
Obviously, it is impossible that the above three statements are true at the same time.
Therefore, the constraints in Equation~\eqref{eq:graph:WkBGPsdp:B} must have been satisfied.

(2) For the constraints in Equation~\eqref{eq:graph:WkBGPsdp:C},
let us consider the set $S$ with a fixed size.
Obviously, the right side of the constraint is a fixed value.
For the left side, considering a given node $u\in{S}$,
the corresponding value of the left side depends on the value of $\frac{1}{2}\|\bm{v}-\bm{u}\|_2^2$.
Then, it is expected that $\frac{1}{2}\|\bm{v}-\bm{u}\|_2^2$ takes a rather small value.
Because $\frac{1}{2}\|\bm{v}-\bm{u}\|_2^2=0$ means that the nodes $u$ and $v$ belong to the same partition,
at most $\frac{n}{k}-1$ nodes can satisfy this condition.
Assume that all nodes satisfying the above condition are placed into $S$, then both sides of Equation~\eqref{eq:graph:WkBGPsdp:C} are 0 and
the constraint is satisfied.
If the size of $S$ increases further, the value of the right side will be increased by 1, at the same time,
the value of the left side will be increased at least by 1.
Therefore, the constraints can be satisfied still.

Then, let us consider the relationship of the optimal values between the semidefinite program and
the {\lxmProblemWordload} problem. For the feasible solution $\{\bm{v}_1,\bm{v}_2,\cdots,\bm{v}_n\}$ constructed, 
the corresponding value of the optimizing goal of Equation~\eqref{eq:graph:WkBGPsdp:A} is represented by $A$.
For any two nodes $v_{i}$ and $v_{j}$, it is not hard to verify that $\mathbf{I}(\textsc{PAR}^*_G(v_{i})\neq\textsc{PAR}^*_G(v_{j}))\equiv\frac{1}{2}\|\bm{v}_{i}-\bm{v}_{j}\|_2^2$,
that is they have the same value indeed.
According to Equation~\eqref{eq:graph:workloadCost} and \eqref{eq:graph:WkBGPsdp:A}, we have $\textsf{cost}_{\Phi}(\textsc{PAR}^*_G)=A$.
Because $A^*$ is the optimal solution of the semidefinite program, obviously, we have $\textsf{cost}_{\Phi}(\textsc{PAR}^*_G)\geq A^*$.
The proof is finished.
\qed
\end{proof}

Theorem~\ref{theorem:graph:WkBGPLowerBound} indicates that 
the optimal solution of Equation~\eqref{eq:graph:WkBGPsdp:A} is a lower bound for the
cost of the optimal partitioning solution.
Then, we need the upper bound.

\begin{lemma}[\cite{Krauthgamer09PartitioningGraphs}]\label{lemma:graph:BGPAlgAppro}
Based on the optimal solution of the semidefinite program, the rounding method $k$-\textsc{Partition}
will return a partitioning solution with approximation ratio $O(\sqrt{\log k\log n})$ in polynomial time,
and the size of each partition can be bounded by $\frac{2n}{k}$.
\qed
\end{lemma}

Then, we have the following result.
\begin{theorem}\label{theorem:graph:WkBGPAppro}
The expected time cost of  {\lxmAlgorithmWorkload} shown in Algorithm~\ref{alg:graph:WkBGPAlg}
is polynomial with the input size, and {\lxmAlgorithmWorkload} is a bi-criteria $(k,2)$-approximation algorithm
with ratio bounded by $O(\sqrt{\log k\log n})$.\qed
\end{theorem}
\begin{proof}
According to Lemma~\ref{lemma:graph:BGPAlgAppro}, Theorem~\ref{theorem:graph:WkBGPLowerBound} and the details
of {\lxmAlgorithmWorkload}, it is easy to verify the correctness of the theorem.
\qed
\end{proof}

\section{Theoretical Analysis of Motif Based Balanced Graph Partitioning Problem}
\label{sec:graph:mkbgptheory}

\subsection{Analysis of Computational Complexities}
In this part, the computational complexities of the {\lxmProblemMotif} problem is analyzed and 
we have the following result.

\begin{theorem}\label{theorem:graph:MkBGPNPC}
The {\lxmProblemMotif} problem is $\lxmNP$-complete, even if $k=2$.\qed
\end{theorem}

Obviously, the proof of Theorem~\ref{theorem:graph:MkBGPNPC} should consider some special motif structures,
that is the parameter $M$ in Definition~\ref{def:graph:MkBGP} needs to be specified.
When the motif $M$ is composed of at most 2 nodes, the {\lxmProblemMotif} problem will be trivial or equivalent
to {\lxmProblemPrevious}, therefore, we use the motif with 3 nodes to analyze the complexities of {\lxmProblemMotif}.
Assuming that the motif $M$ is a simple triangle, we can prove Theorem~\ref{theorem:graph:MkBGPNPCTriangle},
which will imply the lower bound result in Theorem~\ref{theorem:graph:MkBGPNPC}, that is the {\lxmProblemMotif}
problem is $\lxmNP$-hard.

\begin{figure}[t]
	\centering
	\subfigure[An instance $G$ of the graph bisection problem]{
		\label{fig:graph:MkBGPReduction:A}
		\begin{minipage}{.6\textwidth}\centering
			\resizebox{\textwidth}{!}{ %
				\begin{tikzpicture}[ auto,swap]
					\foreach \pos/\name/\vertexID in {{(0,2.5)/a/1}, {(4,4)/b/2}, {(4.5,1)/c/3},
						{(-2,-1)/d/4}, {(2.5,-1.5)/e/5}, {(0,-3)/f/6}, {(5.5,-4)/g/7}, {(9,-3)/h/8}}
					\node[vertex] (\name) at \pos {$v_{\vertexID}$};
					% Connect vertices with edges and draw weights
					\foreach \source/ \dest /\weight in {b/a/v1,d/a/v3,
						b/c/v5, e/c/v6,a/f/v7,
						f/d/v8,f/e/v9,
						g/e/v10,g/h/v13,c/h/17}
					\path[edge] (\source) -- node[pos=0.5,minimum height=10pt,minimum width=0pt](\weight) {} (\dest);
					%	\foreach \nodei/\sangle in {1/0,2/60,3/120,4/180,5/240,6/300}
					\foreach \nodei/\sangle in {1/0,2/36,3/72,4/108,5/144,6/180,7/216,8/252,9/288,10/324}
					\node[fill=blue,circle,minimum size=8pt,inner sep=0pt,outer sep=0pt] (c\nodei) at ($(14 
					,.5)!1!\sangle:(16,.5)$) {};
					\node at (14,3) {\huge$C_3$};
					\foreach \i in {1,2,3,...,10}
					\foreach \j in {1,2,...,\i}
					\path[edge,line width=1pt] (c\i)--(c\j);
					\draw[line width=3pt,loosely dashed,->] ([xshift=5] c.east) to[out=-30,in=140] ([xshift=-5]c4.west);
			\end{tikzpicture}}
	\end{minipage}}
	\\
	\subfigure[An instance $H$ of {\lxmProblemMotifTriangle}]{
		\label{fig:graph:MkBGPReduction:B}
		\begin{minipage}{0.6\textwidth}\centering
			\resizebox{\textwidth}{!}{ %
				\begin{tikzpicture}[ auto,swap]
					\foreach \posx/\posy/\name/\vertexID in {{0/2.5/a/1}, {4/4/b/2}, {4.5/1/c/3}, {-2/-1/d/4}, {2.5/-1.5/e/5}, {0/-3/f/6}, {5.5/-4/g/7}, {9/-3/h/8}}{
						%\node[draw=none] (\name) at (\posx,\posy) {} ;
						\foreach \nodei/\sangle in {1/0,2/36,3/72,4/108,5/144,6/180,7/216,8/252,9/288,10/324}%{1/0,2/60,3/120,4/180,5/240,6/300}
						\node[fill=blue,circle,minimum size=5pt,inner sep=0pt,outer sep=0pt] (\name\nodei) at 
						($(\posx ,\posy)!1!\sangle:(\posx+1.1,\posy)$) {};
						%\node[fill=black,circle,minimum size=8pt,inner sep=0pt,outer sep=0pt] (\name) at ($(\posx ,\posy)$) {};
						%
						\foreach \i in {1,2,3,...,10}
						\foreach \j in {1,2,...,\i}
						\path[edge,line width=.6pt] (\name\i)--(\name\j);
					}
					\foreach \source/ \dest /\weight in {d3/a7/e11,d3/a8/e12,
						b6/a2/e21,b6/a3/e22,
						f3/a10/e31,f3/a9/e32,
						c3/b9/e41,c3/b10/e42,
						h4/c1/e51,h4/c10/e52,
						e2/c7/e61,e2/c8/e62,
						f5/d9/e71,f5/d10/e72,
						f1/e7/e81,f1/e8/e82,
						g3/e1/e81,g3/e6/e82,
						h6/g1/e81,h6/g2/e82}
					\path[edge,line width=1.5pt] (\source) -- node[pos=0.5,minimum height=10pt,minimum width=0pt](\weight) {} (\dest);
					\node at ([xshift=-1cm]a4.west) {\large $C_1$};
					\node at ([xshift=-1cm]c4.west) {\large $C_3$};
					\node at ([xshift=-.8cm]g4.west) {\large $C_7$};
					\node at ([yshift=-2.7cm]f4.west) {\large $C_6$};
					\node at ([xshift=-.8cm]e4.west) {\large $C_5$};
					\node at ([xshift=.5cm]b1.east) {\large $C_2$};
					\node at ([yshift=-2cm]h5.south) {\large $C_8$};
					\node at ([yshift=-2cm]d5.south) {\large $C_4$};
			\end{tikzpicture}}
	\end{minipage}}
\caption{An Example of $\textsc{P}$-time reduction used in the proof of Theorem~\ref{theorem:graph:MkBGPNPCTriangle}}
	\label{fig:graph:MkBGPReduction}
\end{figure}
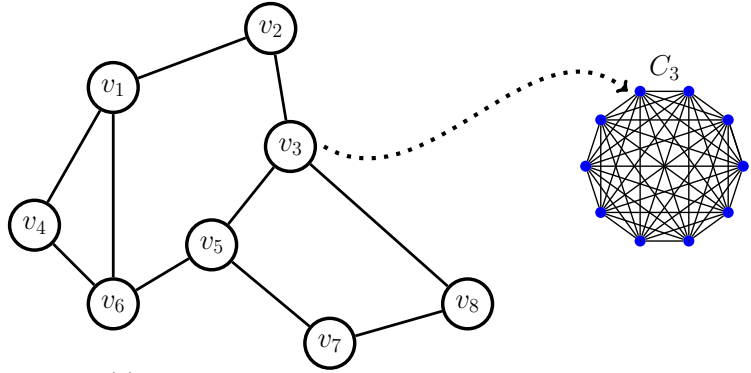
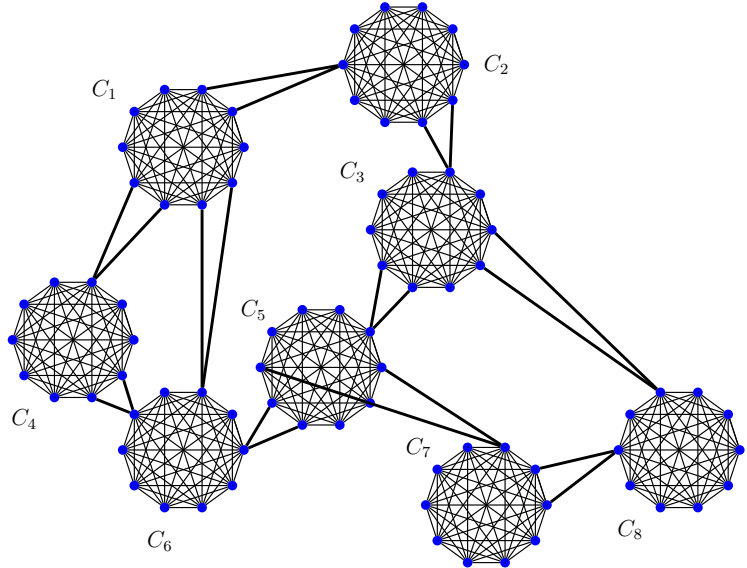

\begin{theorem}\label{theorem:graph:MkBGPNPCTriangle}
Even if $k=2$ and the motif $M$ is a triangle, the {\lxmProblemMotif} problem is still $\lxmNP$-complete.\qed
\end{theorem}
\begin{proof}
We use {\lxmProblemMotifTriangleTwoPar} to represent the special case of {\lxmProblemMotif}
when $k=2$ and the input motif $M$ is a triangle.
Then, the proof can be divided into the following two parts.
(1) The upper bound: {\lxmProblemMotifTriangleTwoPar} is in $\lxmNP$.
(2) The lower bound: {\lxmProblemMotifTriangleTwoPar} is $\lxmNP$-hard.

The proof will utilize the following decision version of {\lxmProblemMotifTriangleTwoPar}.
Given a graph $G$ and an integer $C$, the problem is to determine whether there is a 2-balanced partitioning solution $\textsc{PAR}$
for $G$ such that $\textsf{cost}_{\Delta}(\textsc{PAR})\leq C$,
where $\textsf{cost}_{\Delta}(\cdot)$ is the variant of $\textsf{cost}_{M}(\cdot)$ when $M$ is a triangle.
When the context is clear, we still use {\lxmProblemMotifTriangleTwoPar} to represent the decision version.

For the upper bound, a $\lxmNP$ algorithms for {\lxmProblemMotifTriangleTwoPar} will be introduced to show
that  {\lxmProblemMotifTriangleTwoPar} is in $\lxmNP$.
The algorithm can be explained as follows.
(1) Guess a node set $S\subseteq{V}$ of size $|V|/2$, and initialize a counter \textsf{count} using 0 which is used to compute the number of triangles destroyed by the partitioning solution.
(2) Enumerate all triples $(u,v,w)$ of the node set $V$, if the corresponding three edges $(u,v)$, $(u,w)$ and $(v,w)$ belong to $E$ and 
the three nodes are not all in $S$ or $V\setminus S$, the counter \textsf{count} will be increased by 1.
(3) After all triples have been processed as above, if the counter satisfies $\textsf{count}\leq C$, the answer `yes' will be returned, otherwise, the answer `no' will be returned.

The correctness of the above algorithm can be verified easily, and its time cost can be bounded by $O(|V|^3)$ obviously.
Then, it has been proved that the {\lxmProblemMotifTriangleTwoPar} problem belongs to $\lxmNP$.

For the lower bound, we given a polynomial time reduction form the graph bisection problem to {\lxmProblemMotifTriangleTwoPar}.
Also, the decision version is used, that is, given a graph $G$ and a positive integer $B$, to determine whether there is a 2-balanced partitioning solution
such that the partitioning cost is not larger than $B$.
Then, given an instance $(G,B)$ of the graph bisection problem, the reduction will compute an instance $(H,C)$ of {\lxmProblemMotifTriangleTwoPar}.
The details of the reduction can be explained as follows.
(1) First, an even integer $X$ is computed, such that $X$ satisfies $X\geq 2\cdot\textsf{deg}(G)$ and $X\geq |V_G|+1$
where $\textsf{deg}(G)$ is the maximum degree of the graph $G$.
(2) Then, as shown in Fig.~\ref{fig:graph:MkBGPReduction:A}, for each node $v_i\in{G}$, a complete graph (or clique) $C_i$ is built such that
$C_i$ includes $X$ nodes and there is an edge between any two nodes. Let $H$ be the graph composed of all $C_i$.
(3) Next, several triangles are added into $H$, and the whole procedure will guarantee that each node will be utilized at most once.
For each edge $(v_i,v_j)\in{G}$, assuming that $i<j$, randomly choose two unused nodes $c_{i1}$ and $c_{i2}$ in $C_i$ and one unused
node $c_{j}$ in $C_j$, and add two  edges $(c_{i1},c_{j})$ and $(c_{i2},c_{j})$ to $H$.
As shown in Fig.~\ref{fig:graph:MkBGPReduction}, assume that the input graph $G$ for the bisection problem is shown in Fig.~\ref{fig:graph:MkBGPReduction:A},
the result graph $H$ constructed by the reduction for the {\lxmProblemMotifTriangleTwoPar} problem is shown in Fig.~\ref{fig:graph:MkBGPReduction:B}.
(4) Finally, let the parameter $C$ in {\lxmProblemMotifTriangleTwoPar} be equal to the parameter $B$ in the bisection problem.

Obviously, the above reduction can be finished within $O(|E_G|\cdot X+|V_G|\cdot X^2)$ time,
which can be guaranteed by setting $X$ to be the smallest even number satisfying the conditions, that is $X=O(poly(|G|))$.
Then, the correctness of the reduction can be verified by the following two aspects.

{\ding{239}} If there is a 2-balanced partitioning solution $\textsc{PAR}_G=\{S_G,T_G\}$ for $G$, such that
the corresponding cost is not larger than $B$, we can construct a 2-balanced partitioning solution  $\textsc{PAR}_H=\{S_H,T_H\}$
for $H$ satisfying $\textsf{cost}_{\Delta}(\textsc{PAR}_H)\leq C=B$.
For each node $v_i\in{S_G}$, the construction will add all nodes related with $C_i$ in $H$ to $S_H$,
and for each node $v_j\in{T_G}$, the nodes of $C_j$ in $H$ will be added into $T_H$.
It is not hard to verify that the partitioning solution defined by $\textsc{PAR}_H$ will not destroy
any triangle in $C_i$ but only may destroy the triangles built in the third step of the reduction.
For that step, the nodes utilized are will not be reused during the construction,
so the obtained triangles will not connect with each other.
Meanwhile, if some edge $(v_i,v_j)$ of $G$ is cut by $\textsc{PAR}_G$, then the corresponding triangle in $H$ must
be destroyed also by $\textsc{PAR}_H$, and vice versa.
Therefore, we have $\textsf{cost}_{\Delta}(\textsc{PAR}_H)\leq C=B$.
	
{\reflectbox{\ding{239}}}
If there is a 2-balanced partitioning solution $\textsc{PAR}_H=\{S_H,T_H\}$ for $H$ and
the cost is no larger than $C$, then a  2-balanced partitioning solution $\textsc{PAR}_G=\{S_G,T_G\}$ for $G$ satisfying $\textsf{cost}(\textsc{PAR}_G)\leq B=C$ can be constructed.
First, if $C\geq\frac{|V_G|(|V_G|-1)}{2}$, because $|E_G|\leq \frac{|V_G|(|V_G|-1)}{2}$, then $G$ has a trivial 2-balanced partitioning solution since
any solution will not cause a cost larger than $C$.
Therefore, only the case when $C<\frac{|V_G|(|V_G|-1)}{2}$ is considered in the followings.
	
	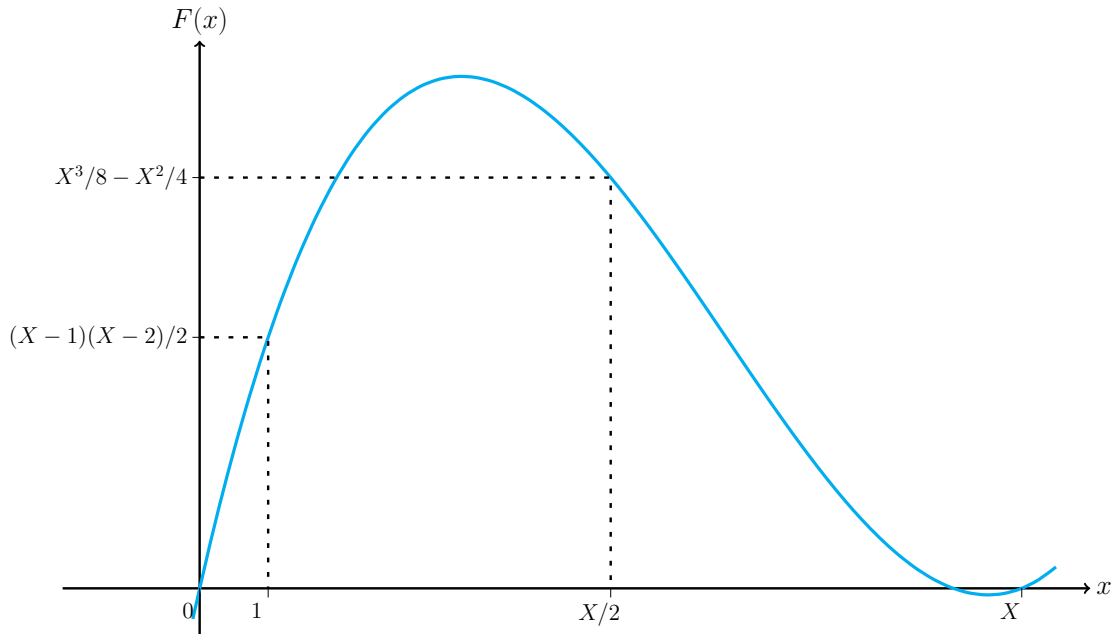
\begin{figure}[ht]
		\centering
		\begin{minipage}{.9\textwidth}\centering
			\resizebox{\textwidth}{!}{ %
				\begin{tikzpicture}[auto,swap,xscale=1.5,yscale=.1]
					\tikzset{elegant/.style={smooth,thick,samples=50,magenta}}
					%\begin{axis}[%axis x line=middle,
					%                 %axis y line=middle,
					%                 ylabel=$f(x)$,
					%                 xlabel=$x$]
					%        \addplot[elegant,domain=-\num:-1/\num]{1/x};
					%        \addplot[elegant,domain=1/\num:\num]{1/x};
					%  \addplot[elegant,orange,domain=0:10]{sin(deg(x))};
					%\addplot[elegant,orange,domain=-1:11]{(10-x)*x*(x-1)/2};
					\draw[->,line width=1.5pt] (-2,0) -- (13,0) node[right] {\LARGE$x$};
					\draw[->,line width=1.5pt] (0,-10) -- (0,120) node[above] {\LARGE$F(x)$};
					%\draw [elegant,orange,domain=-1:10.2,line width=2pt] plot(\x,{(10-\x)*\x*(\x-1)/2});
					\draw [elegant,cyan,domain=-.1:12.5,line width=2pt] plot(\x,{\x*(12-\x)*(11-\x)/2});
					\foreach \x/\xtext in {0/0,1/1,12/X,6/{X/2}}
					\draw (\x,0.1cm) -- (\x,-2cm) node[anchor=north,xshift=-0.25cm] {\Large$\xtext$};
					\foreach \y/\ytext in {55/{(X-1)(X-2)/2},90/{X^3/8-X^2/4}}
					\draw (1pt,\y cm) -- (-3pt,\y cm) node[anchor=east] {\Large$\ytext$};
					\draw[loosely dashed,line width=1.5pt] (0,55) -- (1,55) -- (1,0);
					\draw[loosely dashed,line width=1.5pt] (0,90) -- (6,90) -- (6,0);
					%\end{axis}
			\end{tikzpicture}}
		\end{minipage}
		\caption{An Illustration Example of $F(x)$ used in the proof of Theorem~\ref{theorem:graph:MkBGPNPCTriangle}}
		\label{fig:graph:MkBGPReductionNumplot}
	\end{figure}

In this case, it can be verified that, given an arbitrary $C_i$  in $H$, the nodes in $C_i$ will all belong to $S_H$ or $T_H$, that is the solution $\textsc{PAR}_H$ will
not split $C_i$.
It can be proved by contradiction.
Let us assume that $\textsc{PAR}_H$ split $C_i$.
Then, consider how many triangles in $C_i$ are destroyed.
There are at least $|V_G|+1$ nodes totally in $C_i$, and the two parts obtained by $\textsc{PAR}_H$ are
of size $x$ and $y$ respectively.
Without loss of generality, it is assumed that $x\leq y$.
Then, it can be known that at least $F(x)=x\cdot \frac{y(y-1)}{2}=x\cdot\frac{(X-x)(X-1-x)}{2}$ triangles are destroyed,
where $1\leq x\leq X/2$.
The relationship between $x$ and the function value $F(x)$ is shown intuitively in Fig.~\ref{fig:graph:MkBGPReductionNumplot}.
It is not hard to find that when the condition $1\leq x\leq X/2$ is satisfied, the two extreme values of $F(x)$
are $X^3/8-X^2/4$ (for $x=X/2$) and $(X-1)(X-2)/2$ (for $x=1$).
Furthermore, when $X\geq 2$ which can be satisfied during the construction of the reduction,
the minimum value is obtained by taking $x=1$.
According to the reduction, we have $X\geq |V_G|+1$, therefore, $F(1)=(X-1)(X-2)/2\geq (|V_G|-1)|V_G|/2$.
As a consequence, we have $\textsf{cost}_{\Delta}(\textsc{PAR}_H)\geq F(1)\geq (|V_G|-1)|V_G|/2>C$,
which is a contradiction.
Therefore, it can be known that $\textsc{PAR}_H$ will not split $C_i$, that is all nodes of $C_i$ belong to $S_H$ or $T_H$.

Next, we can use the partitioning solution for $H$ to build a 2-balanced partitioning solution for $G$, whose details are as follows.
For each clique $C_i$ of $H$, if all nodes of $C_i$ belong to $S_H$, the related node $v_i$ of $G$ will be placed into $S_G$.
Otherwise, the nodes of $C_i$ belong to $T_H$ and the corresponding node will be placed into $T_G$.
Then, it can be verified that some edge $(v_i,v_j)$ in $G$ is cut by the partitioning solution $\textsc{PAR}_G$
if and only if the triangle constructed for the edge $(v_i,v_j)$ in $H$ is destroyed by $\textsc{PAR}_H$.
Thus, we have $\textsf{cost}(\textsc{PAR}_G)=\textsf{cost}_{\Delta}(\textsc{PAR}_H)\leq C=B$.

In conclusion, there is a polynomial time reduction from the graph bisection problem to the {\lxmProblemMotifTriangleTwoPar}
problem, and the {\lxmProblemMotifTriangleTwoPar} problem is $\lxmNP$-hard.
Finally, it is proved that the {\lxmProblemMotifTriangleTwoPar} problem is $\lxmNP$-complete.
\qed
\end{proof}

According to Theorem~\ref{theorem:graph:MkBGPNPCTriangle}, the proof of Theorem~\ref{theorem:graph:MkBGPNPC} can be given as follows.
\begin{proof}[The brief proof of Theorem~\ref{theorem:graph:MkBGPNPC}]
The lower bound for the {\lxmProblemMotif} can be obtained trivially,
since, according to Theorem~\ref{theorem:graph:MkBGPNPCTriangle}, when $k=2$
the {\lxmProblemMotifTwo} problem is $\lxmNP$-hard, and it is a special case of {\lxmProblemMotif}.

Therefore, in the followings, all we need to do is to explain that {\lxmProblemMotif} belongs to $\lxmNP$.
The corresponding decision variant of {\lxmProblemMotif} can be expressed as follows.
Given a graph $G$, a positive integer $k$, a $c$-Motif $M$ and a positive integer $C$,
the problem is to determine whether there is a $k$-balanced partitioning solution $\textsc{PAR}_G$ for $G$
such that $\textsf{cost}_{M}(\textsc{PAR}_G)\leq C$.
Then, a $\lxmNP$ algorithm for deciding the {\lxmProblemMotif} problem can be designed as follows.
(1) First, guess a $k$-balanced partitioning solution $\textsc{PAR}_G$ for $G$ and initialize a counter $\textsf{count}$ using 0.
(2) Then, for each pair $(u_1,\cdots,u_c)\in{V_G^c}$, if the induced subgraph obtained on $(u_1,\cdots,u_c)$
contains the same structure as $M$, the process will continue, otherwise, the following steps will be ingored.
Next, it will be checked whether all $\textsc{PAR}_G(u_i)$ ($i\in[1,c]$) are same, and the counter $\textsf{count}$
will be increased by 1 if not same.
(3) Finally, if  $\textsf{count}\leq C$, the answer `yes' will be returned, otherwise, the answer `no' will be returned.

The correctness of the above algorithm can be verified easily.
Since all operations can be implemented within $O(|V_G|^c\cdot c)$ time, the above algorithm is a $\lxmNP$ algorithm for {\lxmProblemMotif}.
Thus, {\lxmProblemMotif} is in $\lxmNP$.

Finally, it has been proved that the {\lxmProblemMotif} problem is $\lxmNP$-complete, even for the special case of $k=2$.
	\qed
\end{proof}

\subsection{Analysis of Inapproximabilities}
According to the above results, it can be known that, unless $\lxmPTime=\lxmNP$,
it is impossible to design polynomial time algorithm for {\lxmProblemMotif}.
The impossibility remains the same even when $k=2$ and $M$ is a triangle.
Then, a natural direction is trying to design efficient approximation algorithms for {\lxmProblemMotif}.
\begin{theorem}
Assume that $\lxmPTime\neq\lxmNP$, even if the motif $M$ is a triangle,
there is no polynomial time algorithms for 
the {\lxmProblemMotif} problem with finite approximation ratio.
\qed
\end{theorem}
\begin{proof}
The proof is similar with \cite{Andreev06BalancedGraph},
which will build a reduction from the 3-Partition problem to {\lxmProblemMotif}.
The input of the 3-Partition problem is composed of $n=3k$ positive integers $\{a_1,\cdots,a_n\}$
and a threshold $S$ satisfying $S/4<a_i<S/2$ and $\sum_{i=1}^{n}a_i=kS$.
The 3-Partition problem is to determine whether there is a partitioning method to group each three
integers and the corresponding sum is exactly $S$.
It is also one of the most classical $\lxmNP$-complete problems\cite{Garey79ComputersIntractablity}.
Fore more, the 3-partition problem is strongly $\lxmNP$-complete, that is, even if the time cost of
the algorithm can be measured using both the input size $n$ and the maximal integer $v_{max}$ as the parameters,
the problem still does not admit polynomial time algorithms unless $\lxmPTime=\lxmNP$.

Using the unary representation of the 3-Partition problem, the following reduction can compute an instance
of {\lxmProblemMotif} within polynomial time.
For each integer $a_i$, construct a clique $K_i$ of size $a_i$, the union of all $K_i$s is just the input graph $H$ of {\lxmProblemMotif}.
Then, it can be proved that, the answer of the 3-Partition problem is `yes' if and only if the reduced {\lxmProblemMotif}
problem admits a $k$-balanced partitioning solution with cost 0.
Assume that there is a polynomial time algorithm $\mathcal{A}$ for {\lxmProblemMotif} with finite approximation ratio,
then we build the input $H$ and invoke the algorithm $\mathcal{A}$ using parameter $k$.
Let the output of $\mathcal{A}$ be $\textsc{PAR}_H$.
If $\textsf{cost}_{\Delta}(\textsc{PAR}_H)=0$, there is a 0-cost partitioning solution for $H$ and then the answer of
the original 3-Partition problem is `yes'.
Otherwise, if $\textsf{cost}_{\Delta}(\textsc{PAR}_H)>0$, since the approximation ratio of $\mathcal{A}$ is finite,
we have $\textsf{cost}_{\Delta}(\textsc{PAR}_{opt})>0$, that is there is no 0-cost partitioning solution for $H$ and the answer of the 3-Partition problem should be `no'.
Combining the reduction algorithm with $\mathcal{A}$, we can obtain an algorithm $\mathcal{B}$,
and $\mathcal{B}$ is a pseudo polynomial time algorithm for 3-Partition, which is a contradiction since it is strongly $\lxmNP$-complete.
Finally, the proof is finished.\qed
\end{proof}

\section{Graph Partitioning Algorithm for Optimizing Motif Computation}
\label{sec:graph:mkbgpalg}

\subsection{SDP Based Approximation Algorithm for {\lxmProblemMotifTriangle}}
First, focusing on the special case that the motif is a triangle,
an approximation algorithm based on SDP is designed for the {\lxmProblemMotif} problem,
and then the proposed algorithm is extended to the general case.

\subsubsection{The Semidefinite Program for {\lxmProblemMotif}}
Given the input graph $G$ and an integer $k$, the semidefinite program for the {\lxmProblemMotifTriangle}
problem can be expressed by Equation~\eqref{eq:graph:sdpMkBGP:goal}.
\begin{alignat}{7}
	{}& \min\quad &{}& \sum_{\bigtriangleup_{\langle a,b,c\rangle}\in G}%w_{\langle a,b,c\rangle}\cdot
	\| \bm{x}_{abc}\|^2_2 &{}& \label{eq:graph:sdpMkBGP:goal}\\
	{}& \mbox{s.t.}\quad&{}& \|\bm{u}-\bm{v}\|_2^2+\|\bm{v}-\bm{w}\|_2^2\;\geq\; \|\bm{u}-\bm{w}\|_2^2 &\quad& \forall u,v,w \in V\label{eq:graph:sdpMkBGP:A}\\
	{}& {}&{}&\sum_{v\in S}\frac{1}{2}\|\bm{u}-\bm{v}\|_2^2\;\geq\;|S|-\frac{n}{k} &{}& \forall S\subseteq V,u\in{S}\label{eq:graph:sdpMkBGP:B}\\
	{}& {}&{}& \|\bm{x}_{abc}\|^2_2\geq \frac{1}{2}\|{\bm{c}-\bm{a}}\|^2_2&{}& \notag\\
	{}& {}&{}& \|\bm{x}_{abc}\|^2_2\geq \frac{1}{2}\|{\bm{a}-\bm{b}}\|^2_2 &{}& \notag \\
	{}& {}&{}& \|\bm{x}_{abc}\|^2_2\geq \frac{1}{2}\|{\bm{b}-\bm{c}}\|^2_2&{}& \forall \bigtriangleup_{\langle a,b,c\rangle}\in G\label{eq:graph:sdpMkBGP:C}
\end{alignat}
Here, similar with Equation~\eqref{eq:graph:WkBGPsdp:A},
the $\|\cdot\|_2^2$ in Equation~\eqref{eq:graph:sdpMkBGP:goal} is the $L2$ normal form.
The constraints represented by Equation~\eqref{eq:graph:sdpMkBGP:A}
and \eqref{eq:graph:sdpMkBGP:B} are the triangle inequality defined over $\ell_2^2$ and the spreading constraints.
The principle of using vectors to represent the partitioning information is also similar with Equation~\eqref{eq:graph:WkBGPsdp:A}.
If the distances between $\bm{u}$ and $\bm{v}$ is too large, it means that $\textsc{PAR}(u)\neq\textsc{PAR}(v)$,
otherwise, $\textsc{PAR}(u)=\textsc{PAR}(v)$.
The goal of the SDP shown in Equation~\eqref{eq:graph:sdpMkBGP:goal}
is to simulate the cost function $\textsf{cost}_{M=\Delta}(\cdot)$ defined by Equation~\eqref{eq:graph:motifCost}
for the {\lxmProblemMotif} problem, which can be explained as follows.
For each triangle $\Delta_{\langle a,b,c\rangle}$ in $G$, the value of the corresponding vector $\bm{x}_{abc}$
indicates whether the triangle $\Delta_{\langle a,b,c\rangle}$ is destroyed in the partitioning solution.
According to the definition of $\textsf{cost}_{M}(\cdot)$ in {\lxmProblemMotif},
the triangle $\Delta_{\langle a,b,c\rangle}$ is destroyed if and only if at least one edge in the triangle is spitted by the partition.
Intuitively, for Equation~\eqref{eq:graph:sdpMkBGP:goal}, the length $\|\bm{x}_{abc}\|_2^2$
of the vector $\bm{x}_{abc}$ is expected to be 0 or 1 in the ideal case.
If $\|\bm{x}_{abc}\|_2^2=1$, the partitioning solution corresponding to the SDP solution will destroy the triangle $\Delta_{\langle a,b,c\rangle}$,
otherwise, all nodes of the triangle $\Delta_{\langle a,b,c\rangle}$ are placed in one partition.
In the practical setting, when building the partitioning solution based on SDP,
if $\|\bm{x}_{abc}\|_2^2$ is small, it will tend to put all the three nodes in one partition, and only when the value of $\|\bm{x}_{abc}\|_2^2$ is not small
it will be possible to split the three nodes.
For constructing the relationship between splitting triangles and the vectors,
the constraints shown in Equation~\eqref{eq:graph:sdpMkBGP:C} are added,
where three constraints for each triangle are added and the total size can be bounded by $O(n^3)$.
According to the discussions above, the expression like $\frac{1}{2}\|\bm{a}-\bm{b}\|_2^2$
can be used to determine whether two nodes should be put together.
Then, according to the constraint $\|\bm{x}_{abc}\|^2_2\geq \frac{1}{2}\|{\bm{a}-\bm{b}}\|^2_2$ defined by
Equation~\eqref{eq:graph:sdpMkBGP:C}, if $\frac{1}{2}\|\bm{a}-\bm{b}\|_2^2$ takes a rather large value, the value of $\|\bm{x}_{abc}\|^2_2$
is large too.
Similarly, only when $\frac{1}{2}\|\bm{a}-\bm{b}\|_2^2=0$, the value of $\|\bm{x}_{abc}\|^2_2$ can be 0
and the optimizing goal \eqref{eq:graph:sdpMkBGP:goal} can be minimized then.
Therefore, because the optimizing goal \eqref{eq:graph:sdpMkBGP:goal} should be minimized,
for the three constraints described in Equation~\eqref{eq:graph:sdpMkBGP:C},
the value of $\|\bm{x}_{abc}\|^2_2$ is 0, if and only if the three constraints satisfy that the values of $\|\bm{x}_{abc}\|^2_2\geq \frac{1}{2}\|{\bm{c}-\bm{a}}\|^2_2$,
$\|\bm{x}_{abc}\|^2_2\geq \frac{1}{2}\|{\bm{a}-\bm{b}}\|^2_2$ and $\|\bm{x}_{abc}\|^2_2\geq \frac{1}{2}\|{\bm{b}-\bm{c}}\|^2_2$ are all 0.
That is, the triangle $\Delta_{\langle a,b,c\rangle}$ is not destroyed.
Otherwise, the value of $\|\bm{x}_{abc}\|^2_2$ will be at least no less than the max value in the three constraints.

Because the size of the constraints defined by Equation~\eqref{eq:graph:sdpMkBGP:C} can be bounded by $O(n^3)$,
then, according to Theorem~\ref{theorem:graph:sdpSolution}, it can be proved that the SDP \eqref{eq:graph:sdpMkBGP:goal}
can be solved within polynomial time.
\begin{proposition}
The SDP shown in Equation~\eqref{eq:graph:sdpMkBGP:goal} can be optimized in polynomial time.\qed
\end{proposition}

\subsubsection{The {\lxmAlgorithmMotif} Algorithm}
In this part, the {\lxmAlgorithmMotif} Algorithm for solving {\lxmProblemMotif}
is introduced, whose main idea is to first solve the corresponding SDP and then use
the rounding techniques to construct the partitioning solution.

\begin{algorithm}[hbt]
	\caption{{\lxmAlgorithmMotif}~(Triangle Computing based $k$ Balanced Graph Partitioning)\label{alg:graph:TrikBGPAlg}}
	\SetKwProg{Fn}{function}{}{end}
	\KwIn{A Graph $G=\{V,E\}$ and an integer $k$.}
	\KwOut{A partition $\textsc{PAR}_G=\{V_1,V_2,\cdots,V_k\}$ for $G$.}
	\Fn{\textsc{Tri-kBGPartition}($G, k$)}{%
		{Compute and maintain all triangles of $G$ into $S_{\Delta}$}\;
		{Solve sdp \eqref{eq:graph:sdpMkBGP:goal} and let $opt=\{\bm{v}_1,\cdots,\bm{v}_n\}$ be the 
			optimal solution}\;
		{$S\leftarrow V$, $\textsc{PAR}_G\leftarrow\emptyset$}\;
		\While{$|S|>\frac{2n}{k}$}{
			{Randomly choose $\bm{r}\in\mathbb{R}^m$ \emph{s.t.} $r_i\sim N(0,1)$ and all $r_i$s are 
				\emph{i.i.d.}}\;
			{$S_r\leftarrow\{v_i\;|\;v_i\in{S},g(\bm{v}_i)\cdot \bm{r}\geq \alpha_{Ck}\}$}\;
			\If{$0<|S_r|\leq\frac{2n}{k}$}{
				{Add $S_r$ to $\textsc{PAR}_G$}\;
				{$S\leftarrow S\setminus S_r$}\;
			}
		}
		\While{$|\textsc{PAR}_G|>k$}{
			{Find the smallest two partition items $V'$ and $V''$ of $\textsc{PAR}_G$}\;
			{Remove $V'$ and $V''$ from $\textsc{PAR}_G$}\;
			{Add $V'\cup V''$ to $\textsc{PAR}_G$}\;
		}
		\textbf{return} {$\textsc{PAR}_G$}\;
	}
\end{algorithm}

The details of {\lxmAlgorithmMotif} are shown in Algorithm~\ref{alg:graph:TrikBGPAlg}.
First, {\lxmAlgorithmMotif} needs to process all triangle structures in $G$ and maintain the
result in the set $S_{\Delta}$ (line 2).
According to Lemma~\ref{lemma:graph:sdp}, {\lxmAlgorithmMotif} can solve the SDP shown
in Equation~\eqref{eq:graph:sdpMkBGP:goal} within polynomial time by invoking a sophisticated solver (line 3),
and let the optimized solution obtained be $\{\bm{v}_1,\cdots,\bm{v}_n\}$.
Then, the set $V$ is used to initialize $S$ and the partition result $\textsc{PAR}_G$ is set to be empty (line 4).
Next, using the rounding technique of \cite{Krauthgamer09PartitioningGraphs}, the
partitioning solution can be built as follows.
\begin{itemize}
	\item[$\checkmark$] 
	First, {\lxmAlgorithmMotif} builds a vector $\bm{r}$ randomly and $\alpha_{Ck}$ based on a variable $X$ generated by
	a normal distribution $N(0,1)$, where the value of $\alpha_{m}$ satisfies $\Pr[X\geq\alpha_{m}]=\frac{1}{m}$ (line 6).
	
	\item[$\checkmark$] 
	Then, {\lxmAlgorithmMotif} utilizes a function $g(\cdot)$ to transform the $\ell_2^2$ distance to a function on $\ell_2$ (line 7).
	The transformation satisfies the following conditions.
	As shown by \cite{Chlamtac06HowPlay}, given $n$ vectors which satisfy the triangle inequality defined on $\ell_2^2$,
	there exists a constant $\delta,A>0$ and  for any $\Lambda>0$ we can find a function $g_{\Lambda}:\mathbb{R}^m\rightarrow\mathbb{R}^n$
	such taht for any two vectors $\bm{u}$ and $\bm{v}$, we have 
	(i) $\|g_{\Lambda}(\bm{u})-g_{\Lambda}(\bm{v})\|_2\leq\frac{A\sqrt{\log 
			n}}{\Lambda}\|\bm{u}-\bm{v}\|_2^2$,
	(ii) 
	$\|\bm{u}-\bm{v}\|_2^2\geq\Lambda\Rightarrow\|g_{\Lambda}(\bm{u})-g_{\Lambda}(\bm{v})\|_2\geq\delta$ and
	(iii) $\|g_{\Lambda}(\bm{u})\|_2=1$.
	Here, the function $g(\cdot)$ used by {\lxmAlgorithmMotif} is a special case of $g_{\Lambda}(\cdot)$ when $\Lambda=2/3$.

	\item[$\checkmark$] 
	Using the above two techniques, in each iteration, {\lxmAlgorithmMotif}
	checks whether the size of $S$ is larger than the $\frac{2n}{k}$ (line 8).
	If the size of $S$ is still too large, {\lxmAlgorithmMotif} randomly selects a vector $\bm{r}$,
	and for each vector $\bm{v}_i$ in $S$, compute the inner product of $\bm{r}$ and $g(\bm{v}_i)$.
	If the result is larger than $\alpha_{Ck}$, the node $v_i$ will be added into $S_r$.
	
	\item[$\checkmark$] 
	During each iteration, if the size of the selected $S_r$ is no larger than $\frac{2n}{k}$,
	it will be treated as a feasible partition, and the set $S$ and $\textsc{PAR}_G$ will be updated (line 8-10).
	Obviously, if each iteration can extract one feasible partition, then the iterations defined between line 6 and 10
	of {\lxmAlgorithmMotif} can be terminated after being executed for polynomial times.
	However, the worst case is that the size of $S_r$ selected in each iteration is too large,
	in this case, the running time can not be bounded efficiently.
	Therefore, only the expected time of {\lxmAlgorithmMotif} is bounded in the following analysis.
	According to \cite{Krauthgamer09PartitioningGraphs}, the probability that the size of $S_r$
	is too large will be not larger than $1/2$.
	Therefore, the expected time of {\lxmAlgorithmMotif} is polynomial.
\end{itemize}

After the rounding procedure, in fact, {\lxmAlgorithmMotif} can not always guarantee
that there are exactly $k$ partitions obtained.
Therefore, {\lxmAlgorithmMotif} needs a post-process to compute the final $k$ partitions (line 11-14).
If the number of partitions is larger than $k$, {\lxmAlgorithmMotif} selects the smallest two partitions
$V'$ and $V''$ from $\textsc{PAR}_G$, and merges them into one partition.
It is not hard to verify that the size of the merged partition will not be larger than $\frac{2n}{k}$.
{\lxmAlgorithmMotif} repeats the above procedure on $\textsc{PAR}_G$ until $|\textsc{PAR}_G|=k$.
Finally, the result $\textsc{PAR}_G$ will be returned by {\lxmAlgorithmMotif} (line 15).

\subsubsection{Approximation Ratio of the \textsc{Tri-kBGPartition} Algorithm}

In this part, the approximation ratio of {\lxmAlgorithmMotif} is analyzed,
which will study the relationship between the optimized solutions of SDP and {\lxmProblemMotif} first,
and then analyze the ratio between the result returned by {\lxmAlgorithmMotif} and
the optimized solution of SDP.

\begin{theorem}\label{theorem:graph:MkBGPLowerBound}
Given an input graph $G$ and an integer $k$, for the {\lxmProblemMotifTriangle} problem,
the optimal solution of the SDP shown in \eqref{eq:graph:sdpMkBGP:goal} is a lower bound
of the optimal partitioning cost of {\lxmProblemMotifTriangle}.
\qed
\end{theorem}
\begin{proof}
The details are omitted due to the space limits.
\end{proof}

According to \cite{Krauthgamer09PartitioningGraphs}, we have  the following result.
\begin{lemma}[\cite{Krauthgamer09PartitioningGraphs}]\label{lemma:graph:preRatio}
Given two nodes $u$ and $v$ in $G$, the probability that {\lxmAlgorithmMotif} will place them to two
different partitions is not larger than $3AB\sqrt{2\log(2Ck)\log n}\cdot||\bm{u}-\bm{v}||_2^2$,
where the parameters $A$, $B$ and $C$ are constants whose details can be found in \cite{Krauthgamer09PartitioningGraphs}.
\qed
\end{lemma}

Given a triangle $\Delta_{\langle a,b,c\rangle}$ in $G$, we use the notation
$\|\Delta_{\langle 	a,b,c\rangle}\|_2^2$ to represent the value of $\|\bm{a}-\bm{b}\|_2^2+\|\bm{b}-\bm{c}\|_2^2+\|\bm{a}-\bm{c}\|_2^2$.
Let $\mathcal{A}_{abc}$ be the event that the triangle $\Delta_{\langle a,b,c\rangle}$ is destroyed  by {\lxmAlgorithmMotif}.
By analyzing the probability of $\mathcal{A}_{abc}$, we have the following result.

\begin{theorem}\label{theorem:graph:roundingboundA}
For a node subset $U\subseteq V$, which satisfies that the three nodes of $\Delta_{\langle a,b,c\rangle}$ are all in $U$,
assuming that during some  iteration of {\lxmAlgorithmMotif} we have $S=U$, the probability that the event $\mathcal{A}_{abc}$
happens satisfies the following condition.
	\begin{equation}
		\Pr[\mathcal{A}_{abc}|S=U]\leq \frac{2\|\Delta_{\langle 
				a,b,c\rangle}\|_2^2}{3}\cdot 3AB\sqrt{2\log(2Ck)\log n}
	\end{equation}
	%\[\Pr[\mathcal{A}_{abc}|S=U]\leq \frac{2\|\Delta_{\langle 
			%a,b,c\rangle}\|_2^2}{3}\cdot 3AB\sqrt{2\log(2Ck)\log n}\]
Here, the parameters $A$, $B$ and $C$ are same as Lemma~\ref{lemma:graph:preRatio}.
\qed
\end{theorem}
\begin{proof}
	Let $\mathcal{A}_{x,y}$ be the event that $x$ and $y$ are separated.
	Assume that $||a-b||_2^2\leq||b-c||^2_2\leq||a-c||^2_2$.
	Obviously, the probability of $\mathcal{A}_{abc}$ can be decomposed into two parts as follows.
	\[\Pr[\mathcal{A}_{abc}]=\Pr[\mathcal{A}_{a,b}]+\Pr[\overline{\mathcal{A}}_{a,b}\wedge\mathcal{A}_{b,c}]\]
	Here, $\overline{\mathcal{A}}$ is the event that $\mathcal{A}$ does not happen.
	Then, we have the following result.
	\begin{align*}
		\Pr[\mathcal{A}_{abc}] & =\Pr[\mathcal{A}_{a,b}]+\Pr[\overline{\mathcal{A}}_{a,b}\wedge\mathcal{A}_{b,c}]\\
		&\leq \Pr[\mathcal{A}_{a,b}]+\Pr[\mathcal{A}_{b,c}]\\
		&\leq \frac{(1+\epsilon)AB}{\epsilon}\sqrt{2\log(Ck/\epsilon)\log n}\bigg(||a-b||_2^2+||b-c||_2^2\bigg)\\
		&\leq \frac{2}{3}\frac{(1+\epsilon)AB}{\epsilon}\sqrt{2\log(Ck/\epsilon)\log n}\bigg(||a-b||_2^2+||b-c||_2^2+||a-c||_2^2\bigg)\\
		&=\frac{2||\bigtriangleup_{\langle a,b,c\rangle}||_2^2}{3}\cdot\frac{(1+\epsilon)AB}{\epsilon}\sqrt{2\log(Ck/\epsilon)\log n}
	\end{align*}
	Here, the second inequality is because of Lemma~\ref{lemma:graph:preRatio},
	the third inequality is derived from the fact that the distance between $a$ and $c$
	is maximized in the triangle $\bigtriangleup_{\langle a,b,c\rangle}$.
\qed
\end{proof}

Based on Theorem~\ref{theorem:graph:MkBGPLowerBound} and \ref{theorem:graph:roundingboundA},
the approximation ratio of {\lxmAlgorithmMotif} can be obtained by the following theorem.

\begin{theorem}[Approximation Ratio of {\lxmAlgorithmMotif}]\label{theorem:graph:MkBGP:ratio}
	The triangle based balanced graph partitioning problem can be solved by a \emph{bi-criteria} approximation algorithm within expected ratio $O(\sqrt{\log k\log n})$,
	where the size of each partition is at most $\frac{2n}{k}$. 
	\qed
\end{theorem}
\begin{proof}
Given an instance $I=\langle G,k\rangle$ of {\lxmProblemMotifTriangle},
let $\textsc{PAR}_G$ be the partitions built by Algorithm~\ref{alg:graph:TrikBGPAlg} and $\textsc{PAR}^*_G$ be
the $k$-balanced partition of $I$ with the optimal partitioning cost.
For more, let $\textsf{cost}_{\Delta}(\textsc{PAR}_G)$ and $\textsf{cost}_{\Delta}(\textsc{PAR}^*_G)$ be the corresponding
partitioning cost of $\textsc{PAR}_G$ and $\textsc{PAR}^*_G$.
	
Since $\textsc{PAR}^*_G$ is optimal, we have $\textsf{cost}_{\Delta}(\textsc{PAR}_G)\geq \textsf{cost}_{\Delta}(\textsc{PAR}^*_G)$
obviously. Then, since the partitioning solution of {\lxmAlgorithmMotif} is based on the
optimal solution of the SDP shown in \eqref{eq:graph:sdpMkBGP:goal}, let
the optimal solution of the SDP be $S^*$.
Then, according to Theorem~\ref{theorem:graph:MkBGPLowerBound}, we have $S^*\leq \textsf{cost}_{\Delta}(\textsc{PAR}^*_G)$.
Therefore, $S^*\leq \textsf{cost}_{\Delta}(\textsc{PAR}^*_G)\leq \textsf{cost}_{\Delta}(\textsc{PAR}_G)$.
	
Because the rounding techniques used by {\lxmAlgorithmMotif} is randomized,
we analyze the expected cost of the partitioning solution returned by {\lxmAlgorithmMotif}.
	\begin{align*}
		\mathbf{E}[\textsf{cost}_{\Delta}(\textsc{PAR}_G)] & =\sum_{\Delta_{\langle 
				a,b,c\rangle}\in{G}}\bigg(\Pr[\mathcal{A}_{abc}]\bigg)\\
		&\leq \sum_{\Delta_{\langle a,b,c\rangle}\in{G}}\bigg(\frac{2\|\Delta_{\langle 
				a,b,c\rangle}\|_2^2}{3}\cdot 3AB\sqrt{2\log(2Ck)\log n}\bigg)\\
		&\leq \sum_{\Delta_{\langle 
				a,b,c\rangle}\in{G}}\bigg(4\|\bm{x}_{abc}\|_2^2\cdot 3AB\sqrt{2\log(2Ck)\log
			n}\bigg)\\
		&=\bigg(12AB\sqrt{2\log(2Ck)\log n}\bigg)\cdot S^*
	\end{align*}
Here, the first inequality is based on Theorem~\ref{theorem:graph:roundingboundA},
the second one is based on the constraints defined by \eqref{eq:graph:sdpMkBGP:C},
and the last one is obtained by considering the definition of the optimizing goal of \eqref{eq:graph:sdpMkBGP:goal}.

Thus, we have the following result.
	\begin{align*}
		S^*&\leq \textsf{cost}_{\Delta}(\textsc{PAR}^*_G)
		\leq\mathbf{E}[\textsf{cost}_{\Delta}(\textsc{PAR}_G)]
		\leq \bigg(12AB\sqrt{2\log(2Ck)\log n}\bigg)\cdot S^*
	\end{align*}
	
Finally, we can obtain the followings.
	\begin{align*}
		\mathbf{E}\bigg[\frac{\textsf{cost}_{\Delta}(\textsc{PAR}_G)}{\textsf{cost}_{\Delta}(\textsc{PAR}^*_G)}\bigg]
		%=\frac{\mathbf{E}[\textsf{cost}_{\Delta}(\textsc{PAR}_G)]}{\textsf{cost}_{\Delta}(\textsc{PAR}^*_G)}
		\leq
		12AB\sqrt{2\log(2Ck)\log n}
		=O\bigg(\sqrt{\log k\log n}\bigg)
	\end{align*}

That is, the expected approximation ratio of {\lxmAlgorithmMotif} can be bounded by $O(\sqrt{\log k\log n})$.
\qed
\end{proof}

\subsection{Extension to the General MkBGP Problem}

Based on the algorithm designed for the {\lxmProblemMotifTriangle} problem, 
for the general case of {\lxmProblemMotif}, {\lxmAlgorithmMotif} can be naturally extended and guarantee the same approximation ratio.
Here, the corresponding SDP can be represented by the following form.
\begin{alignat}{7}
	{}& \min\quad &{}& \sum_{H\in{G\langle M\rangle}}%w_{\langle a,b,c\rangle}\cdot
	\| \bm{x}_{H}\|^2_2 &{}& \label{eq:graph:sdpMkBGPExt:goal}\\
	{}& \mbox{s.t.}\quad&{}& \|\bm{u}-\bm{v}\|_2^2+\|\bm{v}-\bm{w}\|_2^2\;\geq\; \|\bm{u}-\bm{w}\|_2^2 &\quad& \forall u,v,w \in V\label{eq:graph:sdpMkBGPExt:A}\\
	{}& {}&{}&\sum_{v\in S}\frac{1}{2}\|\bm{u}-\bm{v}\|_2^2\;\geq\;|S|-\frac{n}{k} &{}& \forall S\subseteq V,u\in{S}\label{eq:graph:sdpMkBGPExt:B}\\
	%{}& {}&{}& \|\bm{x}_{abc}\|^2_2\geq \frac{1}{2}\|{\bm{c}-\bm{a}}\|^2_2&{}& \notag\\
	%{}& {}&{}& \|\bm{x}_{abc}\|^2_2\geq \frac{1}{2}\|{\bm{a}-\bm{b}}\|^2_2 &{}& \notag \\
	{}& {}&{}& \|\bm{x}_{H}\|^2_2\geq \frac{1}{2}\|{\bm{a}-\bm{b}}\|^2_2&{}& \forall H\in{G\langle M\rangle},(a,b)\in H\label{eq:graph:sdpMkBGPExt:C}
\end{alignat}

Here, the first two constraints used in \eqref{eq:graph:sdpMkBGPExt:goal}  are similarly defined as previous.
The second constraints shown in \eqref{eq:graph:sdpMkBGPExt:C} is an extension of the corresponding constraints used by the {\lxmProblemMotifTriangle} problem.
If $M$ is a motif of size $c$, then for each $H\in{G\langle M\rangle}$, the SDP will includes $c$ conditions.

Obviously, when the size of $M$ is $c$, compared with \eqref{eq:graph:sdpMkBGP:goal},
the SDP shown in \eqref{eq:graph:sdpMkBGPExt:goal} needs $|V|^c$ constraints in addition.
According to Theorem~\ref{theorem:graph:sdpSolution}, it can be still verified that the SDP shown in 
\eqref{eq:graph:sdpMkBGPExt:goal} can be solved within polynomial time.

Finally, the proof of Theorem~\ref{theorem:graph:MkBGPLowerBound} and \ref{theorem:graph:MkBGP:ratio}
can be extended to the general case also, and we have the following result.

\begin{theorem}\label{theorem:graph:MkBGP:generalratio}
Given a $c$-motif $M$, there is a bi-criteria $(k,2)$ approximation algorithm with polynomial time cost for {\lxmProblemMotif},
where the expected approximation ratio can be bounded by $O(\sqrt{\log k\log n})$ and the size of each output partition
is at most $\frac{2n}{k}$.
\qed
\end{theorem}

\section{Conclusion}\label{sec:conclusion}
In this paper, two typical problems of graph partitioning
motivated by the big data computing applications are studied.
For the {\lxmProblemWordload} problem, a bi-criteria polynomial time algorithm
with  $O(\sqrt{\log n\log k})$ approximation ratio is designed.
For the {\lxmProblemMotif} problem, it is proved to be $\lxmNP$-complete and impossible to
be approximated within finite ratio, then a bi-criteria polynomial time algorithm
with  $O(\sqrt{\log n\log k})$ approximation ratio is designed.

\bibliographystyle{plain}
\bibliography{ref}

\end{document}